\documentclass[a4paper,12pt]{article}
\usepackage{amsfonts,slashed}
\usepackage{url}
\usepackage{latexsym}
\usepackage{amsfonts}
\usepackage{epsfig}
\usepackage{latexsym,amssymb}
\usepackage{amsmath,amssymb,amsthm}
\usepackage{ifpdf}
\usepackage{cite}
\usepackage[pdftex]{hyperref}
\usepackage{tikz}
\ifx\pdfoutput\undefined
   \pdffalse
\else
   \pdfoutput=1
   \pdftrue
\fi

\numberwithin{equation}{section}
\def\be{\begin{equation}}
\def\ee{\end{equation}}
\def\ba{\begin{array}}
\def\ea{\end{array}}

\newcommand{\bea}{\begin{eqnarray}}
\newcommand{\eea}{\end{eqnarray}}

\def\ii{{\rm i}}

\newcommand{\bbox}{\lower.2ex\hbox{$\Box$}}

\newtheorem{teo}{Theorem}
\def\bfone{\relax{\rm 1\kern-.35em 1}}
\def\bfzero{\relax{\rm 0\kern -.45 em 0}}
\usetikzlibrary{trees}
\usetikzlibrary{positioning,shadows,arrows,decorations.pathmorphing,decorations.markings}
\tikzstyle{block}=[draw opacity=0.7,line width=1.4cm]

\begin{document}

\title{An Analytic Method for $S$-Expansion involving Resonance and Reduction}
\author{M. C. Ipinza$^{1,2,3}$\thanks{%
marcalderon@udec.cl}, F. Lingua$^{2}$\thanks{%
fabio.lingua@polito.it}, D. M. Pe\~{n}afiel$^{1,2,3}$\thanks{%
diegomolina@udec.cl}, L. Ravera$^{2,3}$\thanks{%
lucrezia.ravera@polito.it} \\
{\small $^{1}$\textit{Departamento de F\'{\i}sica, Universidad de Concepci%
\'{o}n,}} \\
{\small Casilla 160-C, Concepci\'{o}n, Chile}\\
{\small $^{2}$\textit{DISAT, Politecnico di Torino}}\\
{\small Corso Duca degli Abruzzi 24, I-10129 Torino, Italia}\\
{\small $^{3}$\textit{Istituto Nazionale di Fisica Nucleare (INFN)}}\\
{\small Sezione di Torino, Via Pietro Giuria 1, 10125, Torino, Italia}}
\maketitle

\vskip 1 cm

\begin{center}
{\small \textbf{Abstract} }
\end{center}

In this paper we describe an \textit{analytic method} able to give the multiplication table(s) of the set(s) involved in an $S$-expansion process (with either resonance or $0_S$-resonant-reduction) for reaching a target Lie (super)algebra from a starting one, after having properly chosen the partitions over subspaces of the considered (super)algebras.

This analytic method gives us a simple set of expressions to find the subset decomposition of the set(s) involved in the process. Then, we use the information coming from both the initial (super)algebra and the target one for reaching the multiplication table(s) of the mentioned set(s). Finally, we check associativity with an auxiliary computational algorithm, in order to understand whether the obtained set(s) can describe semigroup(s) or just abelian set(s) connecting two (super)algebras.

We also give some interesting examples of application, which check and corroborate our analytic procedure and also generalize some result already presented in the literature.

\vskip 1 cm \eject
\numberwithin{equation}{section}


\section{Introduction}

The relation of given Lie (super)algebras among themselves, and in particular the derivation of new (super)algebras from other ones, is a problem of great interest, in both Mathematics and Physics, since it involves the problem of mixing (super)algebras, which is a non-trivial way of enlarging spacetime symmetries. 

One method to connect different (super)algebras is the \textit{expansion} procedure, introduced for the first time in \cite{Hatsuda}, and subsequently studied under different scenarios in \cite{Azca1,Azca2,Azca3}. 
In $2006$, a natural outgrowth of the power series expansion method was proposed (see Ref.s \cite{Iza1,Iza2,Iza3}), which is based on combining the structure
constants of the initial (super)algebra with the inner multiplication law of a discrete set $S$ with the structure of a semigroup, in order to define the Lie bracket of a new $S$-expanded (super)algebra. 
From the physical point of view, several (super)gravity theories have been extensively studied using the $S$-expansion approach, enabling numerous results over recent years (see Ref.s \cite{Iza4, GRCS, CPRS1, Topgrav, BDgrav, CR2, CRSnew, Static, Gen, Ein, Fierro1, Fierro2, Knodrashuk, Artebani, Concha1, Salgado, Concha2, Caroca:2010ax, Caroca:2010kr, Caroca:2011zz, Andrianopoli:2013ooa, Concha:2016hbt, Concha:2016kdz, Concha:2016tms, Durka:2016eun}) in this context.

The \textit{$S$-expansion} method replicates through the elements of a semigroup the structure of the original (super)algebra into a new one. 
The basis of the $S$-expansion consists, in fact, on combining the multiplication law of a semigroup $S$ with the structure constants of a Lie (super)algebra $\mathcal{G}$
\cite{Iza1}; The new Lie (super)algebra obtained through this procedure is called \textit{$S$-expanded (super)algebra}, and it is written as $\mathcal{G}_S= S \otimes \mathcal{G}$. 

There are two facets applicable in the $S$-expansion method, which offer great manipulation on (super)algebras, \textit{i.e.} \textit{resonance} and \textit{reduction}. The role of \textit{resonance} is that of transferring the structure of the semigroup to the target (super)algebra, and therefore to control its structure with a suitable choice on the semigroup decomposition. Meanwhile, \textit{reduction} plays a peculiar role in cutting the (super)algebra properly, thanks to the existence of a zero element in the set involved in the procedure, 
which allows, for example, the In\"on\"u-Wigner contraction (see Ref.s \cite{Inonu1, Inonu2}).

A fundamental task to accomplish in the $S$-expansion is to find the appropriate semigroup connecting two different (super)algebras, but this task involves a non-trivial process, due to the fact that until today there is no analytic procedure to unequivocally derive the semigroup performing the required expansion. 

With this in  mind, in the present work we describe an \textit{analytic method} to find the correct semigroup(s) allowing $S$-expansion (involving either resonance or $0_S$-resonant-reduction) between two different (super)algebras, once the partitions over subspaces have been properly chosen. 

This work is organized as follows: In Section \ref{Review}, we give a review of $S$-expansion, reduction, $0_S$-reduction (and $0_S$-resonant-reduction), and resonance.
In Section \ref{Method}, we develop an analytic procedure to obtain the semigroup(s) multiplication table(s) linking different Lie (super)algebras. 
Then, examples of application are presented in Section \ref{Examples}. 
Section \ref{Comments} contains a summary of our results, with comments and possible developments.
In the Appendix, we give the detailed calculations for reaching the results we have obtained.

\section{Review of $S$-expansion, reduction, $0_S$-reduction (and $0_S$-resonant-reduction), and resonance}\label{Review}

The expansion of a Lie (super)algebra entails finding a new (super)algebra starting from an original one. The so called \textit{$S$-expansion}, that is an incarnation of the \textit{expansion method} described in \cite{Azca1}, involves a finite abelian semigroup $S$ to accomplish this task, and it has the feature of being very simple and direct (see Ref. \cite{Iza1}). 
The $S$-expansion method allows to obtain new Lie (super)algebras starting from an original one by choosing an abelian semigroup leading to \textit{resonant}, \textit{reduced} or \textit{resonant-reduced} subalgebras.

\subsection{$S$-expansion of Lie (super)algebras}

The $S$-expansion procedure consists in combining the structure constants of a Lie (super)algebra $\mathcal{G}$ with the inner multiplication law of a semigroup $S$, to define the Lie bracket of a new, $S$-expanded (super)algebra $\mathcal{G}_S= S \otimes \mathcal{G}$. 

Let $S= \lbrace \lambda_\alpha \rbrace$, with $\alpha=1,...,N$, be a finite, abelian semigroup with \textit{two-selector} $K_{\alpha \beta}^{\;\;\;\; \gamma}$ defined by
\begin{equation}\label{kseldef}
K_{\alpha \beta}^{\;\;\;\; \gamma} = \left\{ \begin{aligned} &
1 , \;\;\;\;\; \text{when} \; \lambda_\alpha \lambda_\beta = \lambda_\gamma,
\\ & 0 , \;\;\;\;\; \text{otherwise}. \end{aligned} 
\right.
\end{equation}
Let $\mathcal{G}$ be a Lie (super)algebra with basis $\lbrace T_A \rbrace$ and structure constants $C_{AB}^{\;\;\;\;C}$, defined by the commutation relations 
\begin{equation}
\left[T_A, T_B \right]= C_{AB}^{\;\;\;\;C}\; T_C .
\end{equation}
Denote a basis element of the direct product $S\otimes \mathcal{G}$ by $T_{(A,\alpha)}= \lambda_\alpha T_A$ and consider the induced commutator
\begin{equation}
\left[ T_{(A,\alpha)},T_{(B,\beta)}\right] \equiv \lambda_\alpha \lambda_\beta \left[T_A,T_B \right].
\end{equation}
Then one can show (see Ref. \cite{Iza1}) that the product
\begin{equation}\label{prodsexp}
\mathcal{G}_S= S \otimes \mathcal{G}
\end{equation}
corresponds to the Lie (super)algebra given by
\begin{equation}\label{expandedone}
\left[T_{(A,\alpha)},T_{(B,\beta)}\right]= K_{\alpha \beta}^{\;\;\;\; \gamma} C_{AB}^{\;\;\;\;C}\; T_{(C,\gamma)},
\end{equation}
whose structure constants can be written as
\begin{equation}\label{strconstant}
C_{(A,\alpha)(B,\beta)}^{\;\;\;\;\;\;\;\;\;\;\;\;\;\;\;\;(C,\gamma)}= K_{\alpha \beta}^{\;\;\;\gamma}C_{AB}^{\;\;\;\;C}.
\end{equation}
The product $\left[\cdot,\cdot\right]$ defined in (\ref{expandedone}) is also a Lie product, since it is linear, antisymmetric and it satisfies the Jacobi identity. This product defines a new Lie (super)algebra characterized by $(\mathcal{G}_S,\left[\cdot,\cdot\right])$, which is called \textit{$S$-expanded Lie (super)algebra}.
This implies that, for every abelian semigroup $S$ and Lie (super)algebra $\mathcal{G}$, the (super)algebra $\mathcal{G}_S$ obtained through the product (\ref{prodsexp}) is also a Lie (super)algebra, with a Lie bracket given by (\ref{expandedone}) \footnote{However, as we will show in the present work, there exist some exception in which, in order to reach a target Lie (super)algebra, is not always necessary to use a semigroup, but just an abelian set, since the procedure can be performed without requiring associativity. This is due to the fact that, in that cases, the Jacobi identity is trivially satisfied (each term of the Jacobi identity is equal to zero).}.  

\subsection{Reduced Lie (super)algebras}

In \cite{Iza1}, the authors gave a definition in order to introduce the concept of \textit{reduction} of Lie (super)algebras. It essentially reads as follow:
Let us consider a Lie (super)algebra $\mathcal{G}$ of the form $\mathcal{G}=V_0 \oplus V_1$, where $V_0$ and $V_1$ are two subspaces respectively given by  $V_0=\lbrace T_{a_0} \rbrace$ and $V_1=\lbrace T_{a_1} \rbrace$. When $\left[V_0,V_1 \right]\subset V_1$, that is to say when the commutation relations between generator present the following form
\begin{eqnarray}
\left[T_{a_0},T_{b_0}\right]&=& C_{a_0 b_0}^{\;\;\;\;\;c_0} T_{c_0} + C_{a_0b_0}^{\;\;\;\;\;c_1}T_{c_1}, \label{commreduced1}\\ 
\left[T_{a_0},T_{b_1}\right]&=& C_{a_0 b_1}^{\;\;\;\;\;c_1} T_{c_1}, \label{commreduced2}\\ 
\left[T_{a_1},T_{b_1}\right]&=& C_{a_1b_1}^{\;\;\;\;\;c_0}T_{c_0}+C_{a_1b_1}^{\;\;\;\;\;c_1}T_{c_1}, \label{commreduced3}
\end{eqnarray}
one can show that the structure constants $C_{a_0b_0}^{\;\;\;\;\;c_0}$ satisfy the Jacobi identity themselves, and therefore 
\begin{equation}
\left[T_{a_0},T_{b_0}\right]= C_{a_0b_0}^{\;\;\;\;\;c_0}T_{c_0} 
\end{equation}
itself corresponds to a Lie (super)algebra, which is called \textit{reduced} (super)algebra of $\mathcal{G}$.

In spite of the similarity of the concepts, a reduced algebra does \textit{not}, in general, correspond to a subalgebra (see Ref. \cite{Iza1}).

\subsection{$0_S$-reduction (and $0_S$-resonant-reduction) of $S$-expanded Lie (super)algebras}

The concept of \textit{reduction} of Lie (super)algebras, and in particular \textit{$0_S$-reduction}, was introduced in \cite{Iza1}. It involves the extraction of a smaller (super)algebra from a
given Lie (super)algebra $\mathcal{G}_S$, when certain conditions are met. 

Now, in order to give a review of \textit{$0_S$-reduction}, let us consider an abelian semigroup $S$ and the $S$-expanded (super)algebra $\mathcal{G}_S= S \otimes \mathcal{G}$. When the semigroup $S$ has a \textit{zero element} $0_S \in S$ (in the following, we will adopt the notation $0_S\equiv \lambda_{0_S}$, in order to make clearer the multiplication rules of the semigroup(s) involved in the process), this element plays a peculiar role in the $S$-expanded (super)algebra, as it was shown in \cite{Iza1}. 
In fact, we can split the semigroup $S$ into non-zero elements $\lambda_{i}$, $i=0,...,N$, and a zero element $\lambda_{N+1}=0_S= \lambda_{0_S}$. 
The zero element $\lambda_{0_S}$ is defined as one for which
\begin{equation}
\lambda_{0_S} \lambda_\alpha = \lambda_\alpha \lambda_{0_S} = \lambda_{0_S},
\end{equation}
for each $\lambda_\alpha \in S$.
Under this assumption, we can write $S = \lbrace \lambda_{i}\rbrace \cup \lbrace\lambda_{N+1}=\lambda_{0_S}\rbrace$, with $i = 1, ... ,N$ (here and in the following, the Latin index run only on the non-zero elements of the semigroup $\tilde{S}$). 
Then, the two-selector satisfies the relations
\begin{eqnarray}
K_{i,N+1}^{\;\;\;\;\;\;\;\;\;\; j} &=& K_{N+1,i}^{\;\;\;\;\;\;\;\;\;\; j}=0, \\
K_{i,N+1}^{\;\;\;\;\;\;\;\;\;\; N+1} &=& K_{N+1,i}^{\;\;\;\;\;\;\;\;\;\; N+1}=1, \\
K_{N+1,N+1}^{\;\;\;\;\;\;\;\;\;\;\;\;\;\;\;\;j} &=& 0, \\
K_{N+1,N+1}^{\;\;\;\;\;\;\;\;\;\;\;\;\;\;\;\;N+1} &=& 1,
\end{eqnarray}
which mean, when translated into multiplication rules,
\begin{eqnarray}
\lambda_{N+1}\lambda_i &=& \lambda_{N+1}, \\
\lambda_{N+1}\lambda_{N+1} &=& \lambda_{N+1}.
\end{eqnarray}
Therefore, for $\mathcal{G}_S=S \otimes \mathcal{G}$ we can write the commutation relations
\begin{eqnarray}
\left[T_{(A,i)},T_{(B,j)}\right]&=&K_{ij}^{\;\;\;k}C_{AB}^{\;\;\;\;C}T_{(C,k)} + K_{ij}^{\;\;\; N+1}C_{AB}^{\;\;\;\;C}T_{(C,N+1)}, \\
\left[T_{(A,N+1)},T_{(B,j)}\right]&=&C_{AB}^{\;\;\;\;C}T_{(C,N+1)}, \\
\left[T_{(A,N+1)},T_{(B,N+1)}\right]&=&C_{AB}^{\;\;\;\;C}T_{(C,N+1)}.
\end{eqnarray}
If we now compare these commutation relations with (\ref{commreduced1}), (\ref{commreduced2}), and (\ref{commreduced3}), we clearly see that
\begin{equation}\label{eqredalg}
\left[T_{(A,i)},T_{(B,j)}\right]= K_{ij}^{\;\;\;k}C_{AB}^{\;\;\;\;C}T_{(C,k)}
\end{equation}
are the commutation relations of a \textit{reduced} Lie (super)algebra generated by $\lbrace T_{(A,i)} \rbrace$, whose structure constants are $K_{ij}^{\;\;\;\;k}C_{AB}^{\;\;\;\;C}$. 

The reduction procedure, in this particular case, is equivalent to the imposition of the condition
\begin{equation}
T_{A,N+1}=\lambda_{0_S} T_A = 0.
\end{equation}
We can notice that, in this case, the reduction abelianizes large sectors of the (super)algebra, and that for each $j$ satisfying $K_{ij}^{\;\;\;N+1}=1$ (that is to say $\lambda_{0_S} \lambda_j= \lambda_{N+1}$), we have 
\begin{equation}
\left[ T_{(A,i)},T_{(B,j)}\right]=0.
\end{equation}

The above considerations led the authors of \cite{Iza1} to a definition which essentially reads:
Let $S$ be an abelian semigroup with a zero element $\lambda_{0_S} \in S$, and let $\mathcal{G}_S = S \otimes \mathcal{G}$ be an $S$-expanded (super)algebra. Then, the (super)algebra obtained by imposing the condition
\begin{equation}
\lambda_{0_S} T_A =0
\end{equation}
on $\mathcal{G}_S$ (or on a subalgebra of it) is called \textit{$0_S$-reduced (super)algebra} of $\mathcal{G}_S$ (or of the subalgebra).

When the $0_S$-reduced (super)algebra $\mathcal{G}_{S_{R}}$ presents a structure which is \textit{resonant} with respect to the structure of the semigroup involved in the $S$-expansion process, the procedure takes the name of \textit{$0_S$-resonant-reduction}.



\subsection{Resonant subalgebras for a semigroup}

As we have seen, the $S$-expanded (super)algebra has a fairly simple structure. Furthermore, with the reduction procedure we can arrive to a more interesting (super)algebra, where it is possible to demand some abelian commutators. 

Additionally, there is another way to get smaller (super)algebras from $S \otimes \mathcal{G}$, which strongly depends on the structure of semigroup, that we shall see below.

Let $\mathcal{G}=\bigoplus_{p\in I}V_p$ be a decomposition of $\mathcal{G}$ in subspaces $V_p$, where $I$ is a set of indices. For each $p, q \in I$ it is always possible to define the subsets $i_{(p,q)} \subset I$, such that 
\begin{equation}\label{decomposition}
\left[V_p,V_q\right]\subset \bigoplus_{r\in i_{(p,q)}} V_r,
\end{equation}
where the subsets $i_{(p,q)}$ store the information on the subspace structure of $\mathcal{G}$.

Now, let $S=\bigcup_{p\in I} S_p$ be a subset decomposition of the abelian semigroup $S$, such that
\begin{equation}\label{groupdecomposition}
S_p\cdot S_q \subset \bigcup_{r \in i_{(p,q)}} S_r,
\end{equation}
where the product $S_p \cdot S_q$ is defined as
\begin{equation}
S_p \cdot S_q = \lbrace \lambda_\gamma \mid \lambda_\gamma= \lambda_{\alpha_p}\lambda_{\alpha_q}, \; \text{with} \; \lambda_{\alpha_p}\in S_p, \lambda_{\alpha_q}\in S_q \rbrace \subset S.
\end{equation} 
When such subset decomposition $S =\bigcup_{p\in I} S_p$ exists, then we say that this decomposition is in \textit{resonance} with the subspace decomposition of $\mathcal{G}$, $\mathcal{G}= \bigoplus_{p\in I}V_p$. 

The resonant subset decomposition is crucial in order to systematically extract subalgebras from the $S$-expanded (super)algebra $\mathcal{G}_S = S \otimes \mathcal{G}$, as it was enunciated and proven with the following theorem in Ref. \cite{Iza1} \footnote{This theorem corresponds to ``Theorem IV.2" given in Ref. \cite{Iza1}.}:
\begin{teo}\label{Tres} 
Let $\mathcal{G} = \bigcup_{p\in I} V_p$ be a subspace decomposition of $\mathcal{G}$, with a structure described by equation (\ref{decomposition}), and let $S=\bigcup_{p\in I} S_p$ be a resonant subset decomposition of the abelian semigroup $S$, with the structure given in equation (\ref{groupdecomposition}). Define the subspaces of $\mathcal{G}_S = S \otimes \mathcal{G}$ as
\begin{equation}
W_p=S_p\otimes V_p, \;\;\; p\in I.
\end{equation}
Then, 
\begin{equation}
\mathcal{G}_R=\bigoplus_{p\in I}W_p
\end{equation}
is a subalgebra of $\mathcal{G}_S=S\otimes \mathcal{G}$, called \textit{resonant subalgebra} of $\mathcal{G}_S$.
\end{teo}
The proof of Theorem \ref{Tres} can be found in Ref. \cite{Iza1}. 

\section{Theoretical construction of the analytic method for finding the semigroup(s)}\label{Method}

In this section, we develop an \textit{analytic method} to find the semigroup(s)
involved in the $S$-expansion procedure (with either resonance or $0_S$-resonance-reduction) for moving from an initial Lie (super)algebra to a target one, once the partitions over subspaces of the considered (super)algebras have been properly chosen. 

To this aim, let us consider a finite Lie (super)algebra $\mathcal{G}$, which can be decomposed into $N$ subspaces $V_{A}$, with $A=0,1,...,N-1$, and can be written as their direct sum, namely $\mathcal{G}=\bigoplus _{A}V_{A}$. Then, let us consider a target Lie (super)algebra $\mathcal{G}_{S_{RR}}$ (where the label ``$S_{RR}$" stands for ``$S$-expanded, ($0_S$-)resonant-reduced"), which can analogously be decomposed into $N$ subspaces $\tilde{V}_{A}$, with $A=0,1,...,N-1$, and can be written as their direct sum, namely $\mathcal{G}_{S_{RR}}=\bigoplus _{A}\tilde{V}_{A}$ \footnote{Here and in the following, the quantities with a ``tilde" symbol above will refer to quantities of the target (super)algebra.}. 

Let us also consider an abelian, discrete and finite set $\tilde{S}$, with $P$ elements, including the zero element $\lambda_{0_S}$, which can be decomposed into $N$ subsets $S_A$, $A=0,1,...,N-1$. 

We will denote each of this subsets with $S_{\Delta_A}$, where the composed index $\Delta_A$ expresses both the cardinality (number of elements) of each subsets (capital Greek index, $\Delta$), and the subspace associated (capital Latin index, $A,B,C,...$). The association between the subsets and the (super)algebra subspaces is unique (under the \textit{resonance} condition), and we will see that for each value of $A$ we will have a unique value for the corresponding index $\Delta$. This is the reason why we are using this composite index. 

Thus, let us consider the decomposition of the set $\tilde{S}$ in terms of its subsets:
\begin{equation}
\tilde{S}=\sqcup _{\Delta_A } S_{\Delta_A },
\end{equation}
where with the symbol $\sqcup $ we mean the disjoint union of sets. 

We can now use this general decomposition and perform a \textit{$0_S$-resonant-reduced} process \footnote{A process which involves only resonance would be a simpler one, and it will be briefly treated in the following.}, linking the original Lie (super)algebra $\mathcal{G}$ and the target one $\mathcal{G}_{S_{RR}}$. In this way, we get
\begin{equation}\label{resredG01}
\begin{aligned}
\mathcal{G}_{S_{RR}}& = \tilde{V}_0 \oplus  \tilde{V}_1 \oplus \cdots \oplus \tilde{V}_{N-1} = \\ 
& \\
& = \left(S_{\Delta_0}\otimes V_0\right)\oplus \left(\lbrace{\lambda_{0_S}\rbrace}\otimes V_0\right)\oplus \\
& \\
& \;\;\;\; \oplus \left(S_{\Delta_1}\otimes V_1 \right)\oplus \left(\lbrace{\lambda_{0_S}\rbrace}\otimes V_1\right)\oplus \\
& \\
& \;\;\;\; \oplus \cdots \oplus \left(S_{\Delta_{N-1}}\otimes V_{N-1} \right)\oplus \left(\lbrace{\lambda_{0_S}\rbrace}\otimes V_{N-1}\right),
\end{aligned}
\end{equation}
Since we can factorize the zero element, the above relation can be simply rewritten as
\begin{equation}\label{resredG01simple}
\begin{aligned}
\mathcal{G}_{S_{RR}}& = \tilde{V}_0 \oplus  \tilde{V}_1 \oplus \cdots \oplus \tilde{V}_{N-1} = \\ 
& \\
& = \left[\left(S_{\Delta_0}\otimes V_0\right)\oplus \left(S_{\Delta_1}\otimes V_1\right)\oplus \cdots \oplus \left(S_{\Delta_{N-1}}\otimes V_{N-1}\right)\right] \oplus \left(\lbrace\lambda_{0_S}\rbrace \otimes \mathcal{G} \right).
\end{aligned}
\end{equation}
As we have said above, equation (\ref{resredG01simple}) comes from the study of a $0_S$-resonant-reduced process, which we can be written in a more formal way as
\begin{eqnarray}\label{resredG}
\mathcal{G}_{S_{RR}} &=& \tilde{V}_0 \oplus  \tilde{V}_1 \oplus \cdots \oplus \tilde{V}_{N-1} = \\ \nonumber
&& \\ \nonumber
&= &\left[ \tilde{S}\ominus\left( \sqcup _{\Delta _{A\neq 0}}S_{\Delta
_{A}} \oplus \lambda_{0_S}\right) \right] \otimes V_{0} \oplus \\  \nonumber
&& \\  \nonumber
&&\oplus \left[ \tilde{S}\ominus\left( \sqcup _{\Delta _{A\neq 1}} S_{\Delta
_{A}} \oplus \lambda_{0_S}\right) \right] \otimes V_{1}\oplus \\  \nonumber
&& \\  \nonumber
&&\oplus \cdots \oplus \left[ \tilde{S}\ominus\left( \sqcup _{\Delta _{A\neq N-1}} S_{\Delta
_{A}} \oplus \lambda_{0_{S}}\right) \right]\otimes  V_{N-1} = \\  \nonumber
&& \\  \nonumber
&=&\overset{N-1}{\underset{T=0}{\bigoplus }}\left[ \tilde{S}\ominus\left( \sqcup _{\Delta
_{A\neq T}}S_{\Delta _{A}}\oplus \lambda_{0_S}\right) \right] \otimes
V_{T},
\end{eqnarray}
where we have denoted with $\oplus$ and $\ominus$ the direct sum and subtraction over subsets, respectively.

From expression (\ref{resredG}), taking into account the dimensions of the subspaces involved in the partitions of the considered (super)algebras, the following system of equations arises:
\begin{equation}
\left\{
\begin{aligned}
& \dim \left( \tilde{V}_{0}\right)=\dim \left( V_{0}\right) \left( 
\tilde{P}-1-\sum_{A \neq 0}^{N-1}\Delta _{A}\right),  \\
& \dim \left( \tilde{V}_{1}\right) =\dim \left( V_{1}\right) \left( 
\tilde{P}-1-\sum_{A \neq 1}^{N-1}\Delta _{A}\right),  \\
& \;\;\;\;\;\;\;\;\;\;\;\;\;\;\;\;\;\;\;\; \vdots  \\
& \dim \left( \tilde{V}_{N-1}\right)  =\dim \left( V_{N-1}\right) \left( 
\tilde{P}-1-\sum_{A \neq N-1}^{N-1}\Delta _{A}\right),  \\
& \tilde{P} =\sum_{A}^{N-1}\Delta _{A}+1 , 
\end{aligned} \right.
\end{equation}
where in the expression $\tilde{P} =\sum_{A}^{N-1}\Delta _{A}+1$ we have $\tilde{P} \geq P$ (let us remember that $P$ is the total number of elements of the set $\tilde{S}$), and the $+1$ contribution is given by the presence of the zero element $\lambda_{0_S}$. 

We can rewrite the system above in the following simpler form (which comes directly from relation (\ref{resredG01simple})):
\begin{equation}\label{system}
\left\{
\begin{aligned}
& \dim \left( \tilde{V}_{0}\right)  =\dim \left( V_{0}\right) \left( 
\Delta_{0}\right), \\
& \dim \left( \tilde{V}_{1}\right)  =\dim \left( V_{1}\right) \left( 
\Delta_{1}\right),    \\
& \;\;\;\;\;\;\;\;\;\;\;\;\;\;\;\;\;\;\;\; \vdots  \\
& \dim \left( \tilde{V}_{N-1}\right) =\dim \left( V_{N-1}\right) \left( 
\Delta_{N-1}\right),  \\
& \tilde{P} =\sum_{A}^{N-1}\Delta _{A}+1 .
\end{aligned} \right.
\end{equation}
If this system admits a solution (which, if exists, is \textit{unique}), then
we will immediately know, for construction, that it is possible to reach a $S$-expanded, $0_S$-resonant-reduced (super)algebra $\mathcal{G}_{S_{RR}}$ starting from the initial Lie (super)algebra $\mathcal{G}$ with the considered partition over subspaces, and we will also know the way in which the elements of $\tilde{S}$ are distributed into different subsets, \textit{i.e.} the cardinality of the subsets associated with the subspaces of the initial Lie (super)algebra. 

In fact, knowing the dimensions of the partitions of both the initial and the target (super)algebra, the system (\ref{system}) can be solved with respect to the variables
\begin{equation}\label{solutiontothesystem}
\tilde{P}, \; \Delta_A, \; A=0,..., N-1 ,
\end{equation}
and the solution (\ref{solutiontothesystem}) admits only values in $\mathbb{N}^*$ (the value zero is obviously excluded). 

We can observe that the system (\ref{system}) admits solution if and only if the dimensions of the subspaces of the target (super)algebra are proportional (multiples) to the dimensions of the respective subspaces of the initial one \footnote{This is the reason why, if the system (\ref{system}) admits a solution, this solution is trivially unique.}. 
Furthermore, this system admits solutions only if the number of subspaces in the partition of the target (super)algebra is equal to the number of subspaces in the partition of the starting one. 
These considerations offer a criterion to properly choose a partition over subspaces for both the initial and the target Lie (super)algebras, namely:
\begin{itemize}
\item The number of subspaces in the partition of the target (super)algebra must be equal to that of the starting (super)algebra;
\item The dimensions of the subspaces of the target (super)algebra must be multiples of the dimensions of the respective subspaces of the initial one. 
\end{itemize}
Once these two conditions over the partitions are met, one is able to develop our analytic method and find all the semigroup(s), \textit{with respect to the chosen partitions}, linking the considered (super)algebras.

We can also observe that the system (\ref{system}) can also be solved when considering an $S$-expansion including just a \textit{resonant} processes, since it also contains the subsystem
\begin{equation}\label{systemres}
\left\{
\begin{aligned}
& \dim \left( \tilde{V}_{0}\right) =\dim \left( V_{0}\right) \left( 
\Delta_{0}\right),   \\
& \dim \left( \tilde{V}_{1}\right) =\dim \left( V_{1}\right) \left( 
\Delta_{1}\right),    \\
& \;\;\;\;\;\;\;\;\;\;\;\;\;\;\;\;\;\;\;\; \vdots  \\
& \dim \left( \tilde{V}_{N-1}\right) =\dim \left( V_{N-1}\right) \left( 
\Delta_{N-1}\right),  \\
& \tilde{P} =\sum_{A}^{N-1}\Delta _{A},
\end{aligned} \right.
\end{equation}
in which we can clearly see that we are now considering the variable $\tilde{P} =\sum_{A}^{N-1}\Delta _{A}$ \textit{without} the $+1$ contribution, whose presence was due to the inclusion of the zero element $\lambda_{0_S}$. In this case, the solution to the system (\ref{systemres}) is \textit{unique} again, and the considerations done for the $0_S$-resonant-reduced case still hold.

At this point, we know the cardinality of each of the subsets of the set $\tilde{S}$ involved in the process. Now we can understand something more about the multiplication rules of the set $\tilde{S}$, by studying the adjoint representation of the initial Lie (super)algebra with respect to the partition over subspaces.  

Thus, we construct, for each subspace, the adjoint representation with respect to the subspaces. This construction is based on the association
\begin{equation}
\left[V_A,V_B \right] \subset V_C \;\;\; \longrightarrow \;\;\; (C)_{AB}^{C},
\end{equation}
where the index $A,B,C$ can take the values $0,...,N-1$, and where the matrix $(C)_{AB}^{C}$ give us the adjoint representation over the subspace $A$, which can be written as
\begin{equation}
(C)_{AB}^{C} =
\begin{pmatrix}
(C)_{A0}^{0} & (C)_{A0}^{1} & \cdots & (C)_{A0}^{N-1} \\ 
(C)_{A1}^{0} & (C)_{A1}^{1} & \cdots & \vdots \\ 
\vdots & \vdots & \ddots & \vdots \\ 
(C)_{A \; N-1}^{0} & \cdots & \cdots & (C)_{A \; N-1}^{N-1}
\end{pmatrix} .
\end{equation}
In this way, the matrix $(C)_{AB}^{C}$ is written by exploiting the commutation rules between the different partitions over the subspaces of the initial (super)algebra, and it contains the whole information about these partitions.

From this adjoint-like representation over the subspaces of the initial (super)algebra, we can now write, according to the usual $S$-expansion procedure (as it was done in \cite{Iza1}), the relations
\begin{equation}\label{eqsemigrold}
\begin{aligned}
&\left[ \left(  S_{\Delta _{A}} \otimes V_{A}\right) \oplus \left( \lbrace \lambda_{0_S}\rbrace \otimes V_A \right) ,\left(
 S_{\Delta _{B}} \otimes V_{B}\right)\oplus \left( \lbrace \lambda_{0_S}\rbrace \otimes V_B \right) \right] = \\
& = \left(K_{\left( \Delta _{A}\right) \left( \Delta _{B}\right) }^{\left( \Delta
_{C}\right) }(C)_{AB}^{C}\right) \left[\left( S_{\Delta
_{C}}\otimes V_C\right)\oplus \left( \lbrace \lambda_{0_S}\rbrace \otimes V_C \right)\right] ,
\end{aligned}
\end{equation}
where we have also taken into account the presence of the zero element in the set $\tilde{S}$, since we have considered a $0_S$-resonant-reduction process \footnote{A process involving only resonance would be a simpler one, and it would require a similar (but simpler) analysis, since, in that case, one would relax the reduction condition.}.
Here, the composite index $\Delta_A$, $\Delta_B$, and $\Delta_C$ label, as said before, the cardinality of the different subsets (labeled with the capital Greek index $\Delta$), uniquely associated with the different subspace partitions (labeled with capital Latin index). 
In order to write the relation (\ref{eqsemigrold}), we are also taking into account the following theorem:

\begin{teo}\label{T}
In the $S$-expansion procedure, when the commutator of two generators in the original Lie (super)algebra falls into a linear combination involving more than one generator, all the terms appearing in this resultant linear combination of generators must share the same element of the set $\tilde{S}$ involved in the procedure.
\end{teo} 

\begin{proof}
The demonstration of this theorem can be treated as a proof by contradiction (\textit{reductio ad absurdum}). In fact, if the linear combination of generators were coupled with different elements of the set $\tilde{S}$ involved in the procedure, we would have
\begin{equation}\label{linearcombksel}
\begin{aligned}
& \left[T_{(A,\alpha)}, T_{(B,\beta)} \right] = \left[\lambda_\alpha T_{A},\lambda_\beta T_{B} \right]= \lambda_\alpha \lambda_\beta [T_A,T_B] = \\
& = K_{\alpha \beta}^{\gamma_1}C_{AB}^{\;\;\;\;\;C_1}T_{(C_1,\gamma_1)} + K_{\alpha \beta}^{\gamma_2}C_{AB}^{\;\;\;\;\;C_2}T_{(C_2,\gamma_2)} +\cdots + K_{\alpha \beta}^{\gamma_n}C_{AB}^{\;\;\;\;\;C_n}T_{(C_n,\gamma_n)} = \\
& =  K_{\alpha \beta}^{\gamma_1}C_{AB}^{\;\;\;\;\;C_1} \lambda_{\gamma_1}T_{C_1} + K_{\alpha \beta}^{\gamma_2}C_{AB}^{\;\;\;\;\;C_2}\lambda_{\gamma_2}T_{C_2} + \cdots + K_{\alpha \beta}^{\gamma_n}C_{AB}^{\;\;\;\;\;C_n}\lambda_{\gamma_n}T_{C_n} ,
\end{aligned}
\end{equation}
where $\lbrace \lambda_\alpha,\lambda_\beta,\lambda_{\gamma_1},\lambda_{\gamma_2},..., \lambda_{\gamma_n} \rbrace \in \tilde{S}$, $\lbrace T_A \rbrace \in V_A$, $\lbrace T_B \rbrace \in V_B$, $\lbrace T_{C_1},T_{C_2},..., T_{C_n}\rbrace \in V_C$ \footnote{Here we denote with $V_A$, $V_B$, and $V_C$ the subspaces of the partition over the original Lie (super)algebra.}, and where
\begin{equation}\label{diffksel}
K_{\alpha \beta}^{\gamma_1}\neq K_{\alpha \beta}^{\gamma_2}\neq \cdots \neq K_{\alpha \beta}^{\gamma_n}.
\end{equation}
Equations (\ref{linearcombksel}) and (\ref{diffksel}) would mean that different two-selectors were associated with the same resulting element, and, according to the definition of two-selector given in (\ref{kseldef}), this would imply 
\begin{equation}
\lambda_\alpha \lambda_\beta = \lambda_{\gamma_1}= \lambda_{\gamma_2}= \lambda_{\gamma_n},
\end{equation}
with $\gamma_1 \neq \gamma_2 \neq \cdots \neq \gamma_n$, which would break the uniqueness of the internal composition law of the set $\tilde{S}$. 

But this cannot be true, since the composition law associates each couple of elements $\lambda_\alpha$ and $\lambda_\beta$ with a \textit{unique} element $\lambda_\gamma$ (as we can see in the definition (\ref{kseldef})). Thus, we can conclude that when the commutator of two generators in the original Lie (super)algebra falls into a linear combination involving more than one generator, the terms appearing in this resultant linear combination of generators must be multiplied by the same element.
\end{proof}

Theorem \ref{T} reflects on the commutators involving the subspaces of the partition of the original Lie (super)algebra and the subsets of the set $\tilde{S}$.

In fact, if the subspaces involving the linear combination of generators were coupled with different elements of $\tilde{S}$, we would have 
\begin{equation}\label{linearcombkselS}
\begin{aligned}
& \left[\left(\lbrace\lambda_{\alpha ,\Delta_A}\rbrace\otimes V_A\right)\oplus \left(\lbrace\lambda_{0_S}\rbrace\otimes V_A \right),\left(\lbrace\lambda_{\beta ,\Delta_B}\rbrace\otimes V_B\right)\oplus \left(\lbrace\lambda_{0_S}\rbrace\otimes V_B \right) \right] = \\
& = \left(K^{(\gamma_1, \Delta_C)}_{(\alpha, \Delta_A)(\beta, \Delta_B)}(C)^C_{AB} \right)\left[ \left(\lbrace\lambda_{\gamma_1 ,\Delta_C}\rbrace\otimes V_C\right)\oplus \left(\lbrace\lambda_{0_S}\rbrace\otimes V_C \right)\right] + \\
& + \left(K^{(\gamma_2, \Delta_C)}_{(\alpha, \Delta_A)(\beta, \Delta_B)}(C)^C_{AB} \right)\left[ \left(\lbrace\lambda_{\gamma_2, \Delta_C}\rbrace\otimes V_C\right)\oplus \left(\lbrace\lambda_{0_S}\rbrace\otimes V_C \right)\right] + \\
& \;\;\;\;\;\;\;\;\;\;\;\;\;\;\;\;\;\;\;\; \vdots \\
& + \left(K^{(\gamma_n, \Delta_C)}_{(\alpha, \Delta_A)(\beta, \Delta_B)}(C)^C_{AB} \right)\left[ \left(\lbrace\lambda_{\gamma_n, \Delta_C}\rbrace\otimes V_C\right)\oplus \left(\lbrace\lambda_{0_S}\rbrace\otimes V_C \right)\right] ,
\end{aligned}
\end{equation}
where with $\lambda_{\alpha, \Delta_A}$ we denote an arbitrary element $\lambda_\alpha$ contained in the subset $S_A$, associated with the subspace $V_A$, with cardinality $\Delta$ \footnote{The same notation has been adopted in equation (\ref{linearcombkselS}) for all the other elements of the set $\tilde{S}$.}, and where 
\begin{equation}\label{diffkselS}
K^{(\gamma_1, \Delta_C)}_{(\alpha, \Delta_A)(\beta, \Delta_B)} \neq K^{(\gamma_2, \Delta_C)}_{(\alpha, \Delta_A)(\beta, \Delta_B)} \neq \cdots \neq K^{(\gamma_n, \Delta_C)}_{(\alpha, \Delta_A)(\beta, \Delta_B)}.
\end{equation}
Equations (\ref{linearcombkselS}) and (\ref{diffkselS}) would mean that different two-selectors were associated with the same resulting element, which would break the uniqueness of the internal composition law of the set $\tilde{S}$.
Thus, since the composition law associates each couple of elements in $\tilde{S}$ with a \textit{unique} element of the set $\tilde{S}$, we can finally say that
\begin{equation}
\begin{aligned}
& \left[\left(\lbrace\lambda_{\alpha ,\Delta_A}\rbrace\otimes V_A\right)\oplus \left(\lbrace\lambda_{0_S}\rbrace\otimes V_A \right),\left(\lbrace\lambda_{\beta ,\Delta_B}\rbrace\otimes V_B\right)\oplus \left(\lbrace\lambda_{0_S}\rbrace\otimes V_B \right) \right] = \\
& = \left(K^{(\gamma, \Delta_C)}_{(\alpha, \Delta_A)(\beta, \Delta_B)}(C)^C_{AB} \right)\left[ \left(\lbrace\lambda_{\gamma ,\Delta_C}\rbrace\otimes V_C\right)\oplus \left(\lbrace\lambda_{0_S}\rbrace\otimes V_C \right)\right].
\end{aligned}
\end{equation}

By exploiting the statement of Theorem \ref{T}, we can also say that when the commutator of two generators of the original Lie (super)algebra falls into a linear combination involving more then one generator, the intersection between the subsets of the set $\tilde{S}$ could be a non-empty set, which means that the same element(s) will appear in more than one subset of the set $\tilde{S}$ \footnote{In the example given in Subsection \ref{ospexample}, and in particular in Appendix \ref{osp}, we have used this statement; The reader can find the explicit application of this observation in equations (\ref{lastcomm}) and (\ref{equal}).}.

Furthermore, we can also observe that the $S$-expansion procedure does \textit{not always} reproduce an In\"on\"u-Wigner contraction, and this is due to the fact that in the In\"on\"u-Wigner contraction there are some terms in the commutation relations which can go to zero separately, while this cannot happen when one is dealing with the $S$-expansion, where, in fact, the combination of two-selectors appearing in the left-hand-side of equation (\ref{linearcombksel}) can only give either zero or a single two-selector. 
Thus, one can apply our analytic method with the \textit{exception} of the cases in which the $S$-expansion procedure cannot reproduce the In\"on\"u-Wigner contraction.

We may observe that equation (\ref{eqsemigrold}) can be rewritten in a simpler form (due to the fact that the left hand side produces commutation relations that trivially conduce to the zero element of the set $\tilde{S}$), which reads
\begin{equation}\label{eqsemigr}
\left[ S_{\Delta _{A}}\otimes V_{A} ,
 S_{\Delta _{B}} \otimes V_{B} \right] = \left(K_{\left( \Delta _{A}\right) \left( \Delta _{B}\right) }^{\left( \Delta
_{C}\right) }(C)_{AB}^{C}\right) \left[\left( S_{\Delta
_{C}} \otimes V_C\right)\oplus \left( \lbrace \lambda_{0_S}\rbrace \otimes V_C \right)\right] ,
\end{equation}
so as to highlight the information we need to know about the multiplication rules between the elements in the set $\tilde{S}$.

We can now proceed with the development of our analytic method.
The relation (\ref{eqsemigrold}) gives us a first view on the multiplication rules between the elements of the set $\tilde{S}$, since it tells us the way in which the different subsets of $\tilde{S}$ combine among each other, that is to say
\begin{equation}\label{semigrprodold}
\left(S_{\Delta_A}\cup \lbrace \lambda_{0_S}\rbrace \right)\cdot \left(S_{\Delta_B} \cup \lbrace \lambda_{0_S}\rbrace \right) \subset S_{\Delta_C} \cup \lbrace \lambda_{0_S}\rbrace ,
\end{equation}
where the product ``$\cdot$" is the internal product of the set $\tilde{S}$, and thus between its subsets. According to the relation in (\ref{eqsemigr}), equation (\ref{semigrprodold}) can also be rewritten as
\begin{equation}\label{semigrprod}
S_{\Delta_A}\cdot S_{\Delta_B} \subset S_{\Delta_C} \cup \lbrace \lambda_{0_S}\rbrace .
\end{equation}

We have thus exhausted the information coming from the starting (super)algebra $\mathcal{G}$, and we have gained a first view on the multiplication rules of the elements of the subsets of $\tilde{S}$. Now we can exploit the information coming from the target (super)algebra, in order to fix some detail on the multiplication rules and to build up the whole multiplication table describing the set $\tilde{S}$. This step is based on the following \textit{identification criterion}.  

\subsection{Identification criterion}\label{IdCr}

Until now, we have exploited the information coming from the original (super)algebra. 

It is now necessary to understand the structure of the whole multiplication table of the set $\tilde{S}$. To this aim, the other pieces of information we need to know come from the target Lie (super)algebra.
In fact, since at this point, we already know the composition laws between the subsets of $\tilde{S}$, we can now write the following identification between the $S$-expanded generators of the initial Lie (super)algebra and the generators of the target one: 
\begin{equation} \label{identification}
\tilde{T}_A = T_{A,\alpha} \equiv \lambda_\alpha T_A,
\end{equation}
where $T_A$ are the generators included in the subspace $V_A$ of the starting (super)algebra and $\tilde{T}_A$ are the generators in the subspace $\tilde{V}_A$ of the target (super)algebra, and where $\lambda_\alpha \in \tilde{S}$ is a general element of the set $\tilde{S}$ \footnote{We have to remember that when performing our analytic method, we are just talking about a general set $\tilde{S}$, since we do not know yet whether it is, or is not, a semigroup. However, the final check for associativity will tell us if the set $\tilde{S}$ is or is not a semigroup.}.

We have to perform the identification (\ref{identification}) for each element of the set $\tilde{S}$, associating each element of each subset with the generators in the subspace related to the considered subset, that is to say, in our notation,
\begin{eqnarray}
\tilde{T}_A = \lambda_{(\alpha,\Delta_A )}T_A,
\end{eqnarray}
where $\lambda_{(\alpha, \Delta_A)}\equiv \lambda_{\alpha} \in S_{\Delta_A}$.

We can observe that in the development of our analytic method, we can perform the whole procedure of association and identification \textit{without} affecting the internal structure of the generators of the starting (super)algebra.

With the identification (\ref{identification}), we can link the commutation relations between the generators of the target (super)algebra with the commutation relations of the $S$-expanded ones, and, factorizing the elements of the set $\tilde{S}$, we have the chance of fixing the multiplication relations between these elements. To this aim, we first observe that for the target (super)algebra we can write the commutation relations
\begin{equation}\label{commtarg}
\left[ \tilde{T}_A, \tilde{T}_B \right] =\tilde{C}_{AB}^{\;\;\;\;C} \tilde{T}_C,
\end{equation}
where $\tilde{T}_A$, $\tilde{T}_B$ , and $\tilde{T}_C$ are the generators in the subspaces $\tilde{V}_A$, $\tilde{V}_B$, and $\tilde{V}_C$ of the partition over the target Lie (super)algebra, respectively ($A,B,C \in \lbrace 0,..., N-1\rbrace$). Here, with $\tilde{C}_{AB}^{\;\;\;\;C}$ we denote the structure constants of the target Lie (super)algebra, that is to say $\tilde{C}_{AB}^{\;\;\;\;C} \equiv C_{(A,\alpha)(B,\beta)}^{\;\;\;\;\;\;\;\;\;\;\;\;\;\;\;\;(C,\gamma)}$, in the usual notation.
Then, by following the usual $S$-expansion procedure (see Ref. \cite{Iza1}),
since for the initial (super)algebra we can write
\begin{equation}\label{commstart}
\left[T_A, T_B \right]= C_{AB}^{\;\;\;\;C}\; T_C ,
\end{equation}
where we have adopted the same notation used in the case of the target (super)algebra, and 
where $ C_{AB}^{\;\;\;\;C}$ are the structure constants of the initial Lie (super)algebra, 
we are able to write the relations (\ref{expandedone}). We also report them here for completeness:
\begin{equation}
\left[T_{(A,\alpha)},T_{(B,\beta)}\right]= K_{\alpha \beta}^{\;\;\;\; \gamma} C_{AB}^{\;\;\;\;C}\; T_{(C,\gamma)},
\end{equation}
namely
\begin{equation}
\left[\lambda_\alpha T_A,\lambda_\beta T_B \right] = K_{\alpha \beta}^{\;\;\;\; \gamma} C_{AB}^{\;\;\;\;C}\; \lambda_\gamma T_C,
\end{equation}
where the two-selector is defined by (\ref{kseldef}).

We now write the structure constants of the target (super)algebra in terms of the two-selector and of the structure constants of the starting one, namely, reporting equation (\ref{strconstant}) here for completeness,
\begin{equation}
\tilde{C}_{AB}^{\;\;\;\;C} \equiv C_{(A,\alpha)(B,\beta)}^{\;\;\;\;\;\;\;\;\;\;\;\;\;\;\;\;(C,\gamma)}= K_{\alpha \beta}^{\;\;\;\gamma}C_{AB}^{\;\;\;\;C},
\end{equation}
and we exploit the identification (\ref{identification}) in order to write the commutation relations of the target (super)algebra (\ref{commtarg}) in terms of the commutation relations between the $S$-expanded generators of the starting one, factorizing the elements of the set $\tilde{S}$ out of the commutators. 
In this way, we get the following relations:
\begin{equation}\label{commrelafterid}
\left[\lambda_\alpha T_A, \lambda_\beta T_B \right] = K_{\alpha \beta}^{\;\;\;\; \gamma} C_{AB}^{\;\;\;\;C} \lambda_\gamma T_C, \;\;\; \longrightarrow \;\;\; \lambda_\alpha \lambda_\beta \left[T_A, T_B \right] =K_{\alpha \beta}^{\;\;\;\; \gamma} C_{AB}^{\;\;\;\;C} \lambda_\gamma T_C.
\end{equation}
If we now compare the commutation relations (\ref{commrelafterid}) with the ones of the starting (super)algebra in (\ref{commstart}), we are able to deduce something more about the multiplication rules between the elements of $\tilde{S}$, that is to say:
\begin{equation}
\lambda_\alpha \lambda_ \beta = \lambda_\gamma .
\end{equation}
We have to repeat this procedure for all the commutation rules of the target (super)algebra, in order to sculpt the multiplication rules between the elements of the set $\tilde{S}$. 

We observe that, during this process, the possible existence of the zero element in the set $\tilde{S}$, namely $\lambda_{0_S}$, can play a crucial role, since, in the case in which the commutation relations of the target (super)algebra read
\begin{equation}
\left[ \tilde{T}_A, \tilde{T}_B \right] = 0,
\end{equation}
and at the same time from the initial (super)algebra we have
\begin{equation}\label{nozero}
\left[ T_A, T_B \right] \neq 0,
\end{equation}
putting all together the relations
\begin{equation}
\left[\lambda_\alpha T_A, \lambda_\beta T_B \right] = \lambda_\alpha \lambda_\beta \left[ T_A,  T_B \right]  =  C_{AB}^{\;\;\;\;C} \lambda_\gamma T_C =0
\end{equation}
and (\ref{nozero}), we can conclude that
\begin{equation}
\lambda_\alpha \lambda_ \beta = \lambda_{0_S}.
\end{equation}

Thus, at the end of the whole procedure, we are left with the complete multiplication table(s) describing the set(s) $\tilde{S}$ involved in the $S$-expansion (with either resonance or $0_S$-resonant-reduction) process for moving from an initial Lie (super)algebra to a target one. 

The final step consist in checking that $\tilde{S}$ is indeed an abelian semigroup. This is done by checking the \textit{associativity} of the multiplication table(s) (one of the properties required by a set to be defined as a semigroup is, in fact, the associative property).   

\subsection{A note on associativity}\label{Associativity}

The last step consists in analyzing the associative property of the set $\tilde{S}$.
The check for associativity can be rather tedious if performed by hand, but, fortunately, it can be implemented by means of a simple computational algorithm. 
In fact, by mapping the elements $\lambda_i$ of the set $\tilde{S}$ to the set of the integer numbers $\lambda_i\leftrightarrow i\in\mathbb{N}$, it is possible to store the multiplication table of $\tilde{S}$ as a matrix $M$, in a form in which its elements are given by $\lambda_j\lambda_k=\lambda_i\equiv M(j,k)= i$, where $i$ is the index associated with the element $\lambda_i$. Associativity can now be easily tested by checking that, for any $i$, $j$, and $k$, the following relation holds:
\begin{equation}\label{eq_ass_test}
M(M(i,j),k)=M(i,M(j,k)).
\end{equation}

After the check for associativity, the degeneracy of the multiplication table(s) obtained in the analytic procedure after having applied the identification criterion is fixed, and we are left with one (or more) semigroup(s). 

However, in Section \ref{Examples}, we will develop a particular example of application of our analytic method in which, in order to reach the target Bianchi Type II algebra from the Bianchi Type I algebra, the structure of semigroup is \textit{not} necessary, since the procedure can be performed with abelian set(s), without requiring  associativity. This is due to the fact that, in that case, the Jacobi identities of both the mentioned algebras are trivially satisfied (each term of the Jacobi identities is equal to zero). In this work, we just mention this particular case, generalizing the result presented in the literature (see Ref. \cite{Knodrashuk}).

\section{Examples of application} \label{Examples}

In this section, we give some example of application of the analytic method previously developed.
We start with a simple example involving the Bianchi Type I and the Bianchi Type II algebras, and then we move to more complicated cases.
In particular, the last example presented in this section involves the
supersymmetric Lie algebra $osp(32/1)$ and the
hidden superalgebra underlying $D=11$ supergravity, largely discussed in \cite{D'Auria, Ravera}.

The details of the calculations are treated in the Appendix, while in the following we report and discuss our main results.

\subsection{From the Bianchi Type I algebra (BTI) to the Bianchi Type II one (BTII)}

In the following, we apply the method developed in Section \ref{Method} in order to find the possible semigroup(s) leading from the non-trivial Bianchi Type I algebra (BTI) to the Bianchi Type II (BTII) one. To this aim, we first of all analyze the structures of the initial algebra and of the target one.
The only commutator different from zero for the BTI algebra is
\begin{align}
\left[X_1,X_2\right]=X_1,
\end{align}
where $X_1$ and $X_2$ are the generators of the BTI algebra. For the BTII algebra, instead, we have
\begin{align}
\left[Y_1,Y_2\right]=&0, \\
\left[Y_1,Y_3\right]=&0, \\
\left[Y_2,Y_3\right]=&Y_1,
\end{align}
where $Y_1$, $Y_2$, and $Y_3$ are the generators of the BTII algebra.
The details of the calculations are treated in Appendix \ref{Bianchi}, while in the following we report our results.

Performing the steps described in Section \ref{Method}, we obtain the multiplication tables
\begin{equation}\label{megatabla}
\begin{array}{c|cccc}
 &\lambda_a & \lambda_b & \lambda_c &\lambda_{0_S} \\
\hline
\lambda_a & \lambda_{a,0_S} & \lambda_{0_S}& \lambda_{b}&\lambda_{0_S}\\
\lambda_b & \lambda_{0_S}& \lambda_{a,0_S} & \lambda_{a,0_S}&\lambda_{0_S} \\
\lambda_c & \lambda_b & \lambda_{a,0_S}& \lambda_{a,0_S}&\lambda_{0_S}\\
\lambda_{0_S} &\lambda_{0_S} &\lambda_{0_S} &\lambda_{0_S} &\lambda_{0_S}
\end{array}
\end{equation}
We can now perform the following identification:
\begin{equation}\label{idbianchi}
\lambda_a = \lambda_2, \;\;\; \lambda_b = \lambda_3, \;\;\; \lambda_c=\lambda_1, \;\;\; \lambda_{0_S} = \lambda_4.
\end{equation}
Thus, we can rewrite tables (\ref{megatabla}) as follows (where the elements are written in the usual order):
\begin{equation}\label{megatabla1}
\begin{array}{c|cccc}
 &\lambda_1 & \lambda_2 & \lambda_3 &\lambda_4 \\
\hline
\lambda_1 & \lambda_{2,4} & \lambda_3 & \lambda_{2,4} &\lambda_{4}\\
\lambda_2 & \lambda_{3}& \lambda_{2,4} & \lambda_{2,4} &\lambda_{4} \\
\lambda_3 & \lambda_{2,4} & \lambda_{2,4} & \lambda_{2,4}&\lambda_{4}\\
\lambda_4 &\lambda_{4} &\lambda_{4} &\lambda_{4} &\lambda_{4}
\end{array}
\end{equation}
These are the multiplication tables of the possible sets $\tilde{S}$'s involved in the $S$-expansion,$0_S$-resonant-reduced procedure from the BTI algebra to the BTII one. Here we clearly see that the tables described in (\ref{megatabla1}) also include abelian sets that cannot be defined as semigroup, since they do not possess the associative property. 
Fortunately, for both the BTI and BTII algebras, each term of the Jacobi identity
is equal to zero (thus, the Jacobi identity is trivially satisfied), and thus each possible combination of elements in (\ref{megatabla1}) is valid for describing an expansion procedure involving both resonance and reduction, without the necessity of requiring associativity. Thus, the multiplication tables (\ref{megatabla1}) generalize the result previously obtained in \cite{Knodrashuk}.

We can also perform a last step, in order to find the table(s) in (\ref{megatabla}) which describe semigroup(s). This step consists in exploiting the required property of associativity, in order to fix the degeneracy on the multiplication tables (\ref{megatabla}) and finding the semigroup(s) involved in the process. The calculation is rather tedious to be performed by hand, and we have done it with a computational algorithm. For completeness, here we report only the significant relations for checking associativity by hand and understanding which are the semigroups in (\ref{megatabla}):
\begin{align}
& (\lambda_c \lambda_c ) \lambda_b = \lambda_c (\lambda_c \lambda_b ) \;\;\;\Rightarrow \;\;\; \lambda_c \lambda_b = \lambda_{0_S} , \\
& (\lambda_c \lambda_a ) \lambda_a = \lambda_c (\lambda_a \lambda_a ) \;\;\;\Rightarrow \;\;\; \lambda_a \lambda_a = \lambda_{0_S} , \\
& (\lambda_c \lambda_b ) \lambda_b = \lambda_c (\lambda_b \lambda_b ) \;\;\;\Rightarrow \;\;\; \lambda_b \lambda_b = \lambda_{0_S} .
\end{align}
After having checked associativity, we are thus left with the only degeneracy
\begin{equation}
\lambda_{c} \lambda_{c}= \lambda_{a,0_S}.
\end{equation}
We can now substitute the index $a,b,c,0_S$ with numbers. We perform again the identification (\ref{idbianchi}), and we write the multiplication tables thus obtained in terms of $\lambda_i$, with $i=1,2,3,4$, in the usual order:
\begin{equation}\label{tablaf}
\begin{array}{c|cccc}
 &\lambda_1 & \lambda_2 & \lambda_3 &\lambda_4 \\
\hline
\lambda_1 & \lambda_{2,4} & \lambda_{3}& \lambda_{4}&\lambda_4\\
\lambda_2 & \lambda_{3}& \lambda_{4} & \lambda_{4}&\lambda_4 \\
\lambda_3 & \lambda_{4} & \lambda_{4}& \lambda_{4}&\lambda_4\\
\lambda_4 &\lambda_4 &\lambda_4 &\lambda_4 &\lambda_4
\end{array}
\end{equation}
We observe that the abelian, commutative and associative tables (\ref{tablaf}) include the multiplication table of the semigroup $S_{N2}$ described in \cite{Knodrashuk}, namely
\begin{equation}\label{tablaSN2}
\begin{array}{c|cccc}
&\lambda_1 & \lambda_2 & \lambda_3 &\lambda_4 \\
\hline
\lambda_1 & \lambda_{2} & \lambda_{3}& \lambda_{4}&\lambda_4\\
\lambda_2 & \lambda_{3}& \lambda_{4} & \lambda_{4}&\lambda_4 \\
\lambda_3 & \lambda_{4} & \lambda_{4}& \lambda_{4}&\lambda_4\\
\lambda_4 &\lambda_4 &\lambda_4 &\lambda_4 &\lambda_4
\end{array}
\end{equation}
which represents a possible semigroup for moving form a BTI algebra to a BTII algebra, through a $0_S$-resonant-reduction procedure. The degeneracy appearing in (\ref{tablaf}) (namely $\lambda_1 \lambda_1 = \lambda_{2,4}$) shows us that there are two possible semigroups able to give the same result (one of them is the same described it \cite{Knodrashuk}, $S_{N2}$, while the other one is a new result that we have obtained with our analytic procedure).

We have thus given an example in which the method described in Section \ref{Method} allows us to find the semigroups for moving from the BTI algebra to a $S$-expanded, $0_S$-resonant-reduced one (BTII), once the partitions over subspaces have been properly chosen. 

We can now try to achieve the same result, by considering an $S$-expansion with only a \textit{resonant} structure (relaxing the reduction condition). To this aim, we study the system (\ref{systemres}), which, in this case, is solved by
\begin{equation}
\tilde{P}=3, \;\;\; \Delta_0=1, \;\;\; \Delta_1=2 .
\end{equation}
Then, performing the usual procedure (see Section \ref{Method}) and assuming the same identification presented in the detailed calculations for the previous case (see Appendix \ref{Bianchi}), namely
\begin{align}
\lambda_a X_2=&Y_3 \\
\lambda_b X_1=&Y_1\\
\lambda_c X_1=&Y_2,
\end{align}
we can reach the multiplication rules between the elements of the set $\tilde{S}$, after having faced the particular situation in which
\begin{align}
\left[Y_1,Y_3\right]&=0, \label{commdiff} \\
\left[\lambda_bX_1,\lambda_aX_3\right]&=0,\\
\lambda_b\lambda_a\left[X_1,X_2\right]&=0,
\end{align}
where $\left[X_1,X_2\right]=X_1\neq0$. As we can see in equation (\ref{commdiff}), the generators $Y_1$ and $Y_3$ of the target algebra must commute, while the generators $X_1$ and $X_2$ of the starting algebra do not commute; so, the only way for reaching a consistent multiplication rule between the elements $\lambda_a$ and $\lambda_b$ consists in adding a zero element in the set $\tilde{S}$ involved in the process, such that 
\begin{equation}
\lambda_b \lambda_a = \lambda_{0_S}.
\end{equation}
The inclusion of the zero element is consistent, since this modification just affects the variable $\tilde{P}$ in the system (\ref{systemres}), which increases of $+1$ (namely, $\tilde{P}=4$). In this way, the multiplication table of the set $\tilde{S}$ acquires both a new row and a new column, without affecting associativity, in the case of a semigrup table.

We can thus conclude that, in this case, our analytic method shows us the necessity of including a $0_S$-reduction to the resonant process too. 
Thus, we can reach the following multiplication rules:
\begin{align}
\lambda_a \lambda_a =& \lambda_a , \nonumber \\
\lambda_c \lambda_a =& \lambda_b , \nonumber \\
\lambda_{b,c} \lambda_{b,c} =& \lambda_a , \nonumber \\
\lambda_{b} \lambda_{a} =& \lambda_{0_S},
\end{align}
and the multiplication table, after having performed the identification
\begin{equation}
\lambda_a = \lambda_2, \;\;\; \lambda_b = \lambda_3, \;\;\; \lambda_c=\lambda_1, \;\;\; \text{with the extra zero element} \; \lambda_{0_S}=\lambda_4,
\end{equation}
reads
\begin{equation}\label{onlyrestablefin}
\begin{array}{c|ccc}
& \lambda_1 & \lambda_2 & \lambda_3\\
\hline 
\lambda_1 & \lambda_2 & \lambda_3 & \lambda_2\\
\lambda_2 & \lambda_3 & \lambda_2 & \lambda_4\\
\lambda_3 & \lambda_2 & \lambda_4 & \lambda_2
\end{array}\longrightarrow \begin{array}{c|cccc}
& \lambda_1 & \lambda_2 & \lambda_3 & \lambda_4\\
\hline 
\lambda_1 & \lambda_2 & \lambda_3 & \lambda_2& \lambda_4\\
\lambda_2 & \lambda_3 & \lambda_2 & \lambda_4 & \lambda_4\\
\lambda_3 & \lambda_2 & \lambda_4 & \lambda_2& \lambda_4\\
\lambda_4 & \lambda_4 & \lambda_4 & \lambda_4 & \lambda_4
\end{array}
\end{equation} 
Let us finally observe that table (\ref{onlyrestablefin}), which is abelian (but \textit{not} associative), is included in the multiplication tables (\ref{megatabla1}), previously obtained in the context of $0_S$-resonant-reduction.

\subsection{$iso(2,1)$ from the $0_{S}$-resonant-reduction of $so(2,2)$}\label{soisomain}

In this example, our aim is to find the multiplication table(s) of the semigroup(s) connecting
the Lie algebras $so\left(2,2\right) $ and $iso\left( 2,1\right) $ in three dimensions, through a $0_S$-resonant-reduction process, after having properly chosen the partitions over subspaces.

We can write $so(2,2)=\{J_i,P_i\}$, with $i=1,2,3$, and $iso(2,1)=\{\tilde{J}_i,\tilde{P}_i\}$, with $i=1,2,3$, where we have considered $J^i = \frac{1}{2}\epsilon^{ijk}J_{jk}$, according to the notation used in \cite{Witten}. The commutation relations between the generators of the starting $so(2,2)$ algebra can be simply written as 
\begin{align}
& \left[ J_{i},J_{j}\right] = \epsilon_{ijk}  J^{k}, \\
& \left[ J_{i},P_{j}\right] = \epsilon_{ijk}  P^{k}, \\
& \left[ P_{i},P_{j}\right] = \epsilon_{ijk}  J^{k},
\end{align}%
where $i,j,k,...=1,2,3$, and the commutation relations between the
generators of the target $iso(2,1)$ algebra can be written as 
\begin{align}
& \left[ \tilde{J}_{i},\tilde{J}_{j}\right] = \epsilon_{ijk} \tilde{J}^{k}, \\
& \left[ \tilde{J}_{i},\tilde{P}_{j}\right] = \epsilon_{ijk} \tilde{P}^{k}, \\
& \left[ \tilde{P}_{i},\tilde{P}_{j}\right] = 0,
\end{align}%
where, again, $i,j,k,...=1,2,3$. 
Thus, following the procedure described in Section \ref{Method}, we reach the multiplication table 
\begin{equation}
\begin{array}{c|ccc}
& \lambda _{0} & \lambda _{1} & \lambda _{2} \\ \hline
\lambda _{0} & \lambda _{0} & \lambda _{1} & \lambda _{2} \\ 
\lambda _{1} & \lambda _{1} & \lambda _{2} & \lambda _{2} \\ 
\lambda _{2} & \lambda _{2} & \lambda _{2} & \lambda _{2}%
\end{array}
\label{tabla123}
\end{equation}%
The detailed calculations are treated in Appendix \ref{soiso}, while in the following we discuss our result.

Table (\ref{tabla123}) is an abelian, commutative and associative multiplication table (the check for associativity in this case is simple). Thus, we are left with the semigroup that allows us to move from the Lie algebra $so(2,2)$ to the Lie algebra $iso(2,1)$ one
through $0_{S}$-resonant-reduction, just performing the analytic procedure
described in Section \ref{Method}. We have found out that a single semigroup (with respect to the chosen partitions) is involved in the process, and it corresponds to the well known semigroup $S^{(1)}_E$, which is given by
\begin{equation}
\lambda _{\alpha }\lambda _{\beta }=\left\{ \begin{aligned} &
\lambda_{\alpha+\beta} , \;\;\;\;\; \text{when} \; \alpha+\beta \leqslant 2,
\\ & \lambda_2 , \;\;\;\;\;\;\;\;\; \text{when} \; \alpha + \beta > 2. \end{aligned}%
\right.
\end{equation}
As we can see, table (\ref{tabla123}) perfectly fits this description.

\subsection{The Maxwell algebra ($\mathcal{M}$) as a $0_S$-resonant-reduction of the Anti-de Sitter ($AdS$) Lie algebra}

With the analytic procedure described in Section \ref{Method}, we can find the semigroup linking the Anti-de Sitter ($AdS$) and the Maxwell ($\mathcal{M}$) algebras, through the $S$-expansion ($0_S$-resonant-reduction) procedure. In the following, we will show that it is exactly the one obtained in \cite{Concha1, Salgado, Concha2}, that is to
say the semigroup $S_{E}^{(2)}$, which satisfies the multiplication law 
\begin{equation}
\lambda _{\alpha }\lambda _{\beta }=\left\{ \begin{aligned} &
\lambda_{\alpha+\beta} , \;\;\;\;\; \text{when} \; \alpha+\beta \leqslant 3,
\\ & \lambda_3 , \;\;\;\;\;\;\;\;\; \text{when} \; \alpha + \beta > 3. \end{aligned}%
\right.
\end{equation}

We start with the analysis of the two mentioned algebras. The generators
of the $AdS$ algebra are $\lbrace J_{ab},P_a \rbrace$, and they satisfy the
commutation relations 
\begin{align}
& \left[ J_{ab}, J_{cd}\right] =
\eta_{bc}J_{ad}-\eta_{ac}J_{bd}-\eta_{bd}J_{ac}+\eta_{ad}J_{bc}, \\
& \left[J_{ab},P_c \right] = \eta_{bc}P_a - \eta_{ac}P_b , \\
& \left[P_a, P_b \right] = J_{ab}.
\end{align}
The generators of the Maxwell algebra $\mathcal{M}$ are $\{\tilde{J}_{ab},\tilde{P}_{a},\tilde{Z}_{ab}\}$, and they satisfy the following commutation relations 
\begin{align}
& \left[ \tilde{J}_{ab},\tilde{Z}_{cd}\right] =\eta _{bc}\tilde{Z}_{ad}-\eta
_{ac}\tilde{Z}_{bd}-\eta _{bd}\tilde{Z}_{ac}+\eta _{ad}\tilde{Z}_{bc}, \\
& \left[ \tilde{Z}_{ab},\tilde{P}_{a}\right] =0, \\
& \left[ \tilde{Z}_{ab},\tilde{Z}_{cd}\right] =0, \\
& \left[ \tilde{J}_{ab},\tilde{J}_{cd}\right] =\eta _{bc}\tilde{J}_{ad}-\eta
_{ac}\tilde{J}_{bd}-\eta _{bd}\tilde{J}_{ac}+\eta _{ad}\tilde{J}_{bc}, \\
& \left[ \tilde{J}_{ab},\tilde{P}_{c}\right] =\eta _{bc}\tilde{P}_{a}-\eta
_{ac}\tilde{P}_{b}, \\
& \left[ \tilde{P}_{a},\tilde{P}_{b}\right] =\tilde{Z}_{ab}.
\end{align}
We observe that a particular characteristic of the Maxwell algebra is given
by the relation 
\begin{equation}
\left[ \tilde{P}_{a},\tilde{P}_{b}\right] =\tilde{Z}_{ab},
\end{equation}%
and that $\tilde{Z}_{ab}$ commutes with all generators of the algebra, except
the Lorentz generators $J_{ab}$. 

Interestingly, the Maxwell algebra $\mathcal{M}$ can be obtained with an In\"on\"u-Wigner contraction of the $AdS$-Lorentz algebra \footnote{This can be easily proved by performing on the $AdS$-Lorentz (super)algebra (the supersymmetric extension of the $AdS$-Lorentz algebra is displayed in Ref. \cite{gaussbonnet}) the following redefinition of the generators ${J}_{ab}\rightarrow {J}_{ab}, \; {Z}_{ab}\rightarrow \frac{1}{\bar{e}^2}{Z}_{ab}, \; {P}_{a}\rightarrow \frac{1}{\bar{e}}{P}_{a} , \left(Q_\alpha \rightarrow \frac{1}{\bar{e}}Q_{\alpha} \right)$, which provides us with
the Maxwell (super)algebra $(s)\mathcal{M}$ in the limit $\bar{e}\rightarrow0$ (here we have relaxed, for simplicity, the notation with the symbol ``tilde" above the generators).} (of which the Lorentz type algebra $\mathcal{L}=\lbrace J_{ab},Z_{ab} \rbrace$ is a subalgebra), whose supersymmetric extension was deeply studied in Ref. \cite{gaussbonnet} in the context of the supersymmetry invariance of a supergravity theory in the presence of a non-trivial boundary. 

The details of the calculations involved in this example are reported in Appendix \ref{Maxwell}, while in the following we discuss our main results.

By performing the analytic procedure described in Section \ref{Method}, we reach the multiplication table
\begin{equation}
\begin{array}{c|cccc}
& \lambda _{0} & \lambda _{1} & \lambda _{2} & \lambda _{3} \\ 
\hline
\lambda _{0} & \lambda _{0} & \lambda _{1} & \lambda _{2} & \lambda _{3} \\ 
\lambda _{1} & \lambda _{1} & \lambda _{2} & \lambda _{3} & \lambda _{3} \\ 
\lambda _{2} & \lambda _{2} & \lambda _{3} & \lambda _{3} & \lambda _{3} \\ 
\lambda _{3} & \lambda _{3} & \lambda _{3} & \lambda _{3} & \lambda _{3}
\end{array}
\label{tablaMaxwell}
\end{equation}%
This table represents an abelian, commutative and associative semigroup, that is exactly the well known semigroup $S^{(2)}_E$ found in \cite{Concha1, Salgado, Concha2}.

\subsection{From the supersymmetric Lie algebra $osp(32/1)$ to the hidden superalgebra underlying $D=11$ supergravity}\label{ospexample}

With this example, we move to \textit{superalgebras}, and in particular we concentrate on the supersymmetric Lie algebra $osp(32/1)$ and on the hidden superalgebra underlying supergravity in eleven dimensions.

Simple supergravity in $D=11$ was first constructed in \cite{Cremmer}.
The bosonic field content of $D=11$ supergravity is given by the metric $g_{\mu\nu}$ and by a $3$-index antisymmetric tensor $A_{\mu\nu\rho}$ ($\mu,\nu,\rho,...=0,1,...,D-1$); The theory also presents a single Majorana gravitino $\Psi_\mu$ in the fermionic sector. 
By dimensional reduction (as it was shown in \cite{Cremmer2}), the theory yields $\mathcal{N}=8$ supergravity in four dimensions, which is considered a possibly viable unification theory of all interactions. 

An important task to accomplish was the identification of the supergroup underlying the theory, and allowing the unification of all elementary particles in a single supermultiplet, since a supergravity theory whose supergroup is unknown is an incomplete one. 

The need for a supergroup was already felt by the inventors of the theory, and in \cite{Cremmer} the authors proposed $osp(32/1)$ as the most likely candidate. However, the field $A_{\mu \nu \rho}$ of the Cremmer-Julia-Scherk theory is a $3$-form rather than a $1$-form, and therefore it cannot be interpreted as the potential of a generator in a supergroup.

The structure of this same theory was then reconsidered in \cite{D'Auria, Ravera},
in the Free Differential Algebra (FDA) framework, using the superspace
geometric approach. In \cite{D'Auria}, the supersymmetric FDA was also
analyzed in order to see whether the FDA formulation could be interpreted in
terms of an ordinary Lie superalgebra (in its dual Maurer-Cartan
formulation), introducing the notion of Cartan integrable systems. 
This was proven to be true, and the existence of a hidden superalgebra underlying the theory was presented for the first time (the authors got a dichotomic solution, consisting in two different supergroups, whose $1$-form potentials can be alternatively used to parametrize the $3$-form).

This hidden superalgebra includes, as a subalgebra, the super-Poincar\'{e}
algebra of the eleven-dimensional theory, but it also involves two extra bosonic
generators $Z^{ab}, Z^{a_1 \cdots a_5}$ ($a,b,\cdots = 0,1,\cdots 10$),
commuting with the $4$-momentum $P_a$ and having appropriate commutators
with the $D=11$ Lorentz generators $J_{ab}$. The generators that commute
with all the superalgebra but the Lorentz generators can be named ``almost
central". 

Furthermore, to close the algebra, an extra nilpotent fermionic
generator $Q^{\prime }$ must be included. In the following, we will replace
the notation in \cite{D'Auria, Ravera} as follows 
\begin{align}
Z^{ab}\rightarrow &\tilde{Z}^{ab} , \\
Z^{a_1...a_5}\rightarrow&\tilde{Z}^{a_1,...a_5}, \\
Q^{\prime }\rightarrow &\tilde{Q}',
\end{align}
in order to be able to recognize the generators of the target
superalgebra from the generators of starting one, as we have previously done along the paper.

The bosonic generators $Z^{ab}$ and $Z^{a_1 \cdots a_5}$ were understood as $p$-brane charges, sources of dual potentials \cite{Hull, Townsend}.
The role played by the extra fermionic generator $Q'$ was much less investigated, and the most relevant contributions were given first in \cite{vanHolten}, and then in particular in \cite{Bandos}, where the results in \cite{D'Auria} were further analyzed and generalized. 

Recently, in \cite{Ravera}, the authors have shown that, as the generators
of the hidden super Lie algebra span the tangent space of a supergroup
manifold, then, in the geometrical approach, the fields are naturally defined
in an enlarged manifold, corresponding to the supergroup manifold, where all
the invariances of the FDA are diffeomorphisms, generated by Lie
derivatives. 

The extra spinor $1$-form involved in the construction of the
hidden superalgebra allows, in a dynamical way, the diffeomorphisms in the
directions spanned by the almost central charges to be particular gauge
transformations, so that one obtains the ordinary superspace as the quotient
of the supergroup over the fiber subgroup of gauge transformations.

We now want to show that, with the analytic method developed in Section \ref{Method}, we are able to find the semigroup which is involved in the $S$-expansion ($0_S$-resonant-reduction) procedure for moving from the original $osp(32/1)$ Lie algebra to the hidden superalgebra underlying supergravity in eleven dimensions. 

This achievement tell us that the method described in \cite{D'Auria, Ravera}, which is based on the development of the FDA in terms of $1$-forms (the Maurer-Cartan formulation of the FDA has a dual description in terms of commutation relations of the considered Lie algebra, as it is shown in \cite{Iza3}), lead to the same result (that is to say, to the same hidden superalgebra) that can be found performing a $S$-expansion ($0_S$-resonant-reduction) procedure from $osp(32/1)$, with an appropriate semigroup. We will display the multiplication table of the mentioned
semigroup in the following, and we will see that it is the semigroup $S^{(3)}_E$, which satisfies the multiplication rules 
\begin{equation}
\lambda_\alpha \lambda_\beta = \left\{ \begin{aligned} &
\lambda_{\alpha+\beta} , \;\;\;\;\; \text{when} \; \alpha+\beta \leqslant 4,
\\ & \lambda_4 , \;\;\;\;\;\;\;\;\; \text{when} \; \alpha + \beta > 4. \end{aligned} 
\right.
\end{equation}
The same result was previously achieved in \cite{Iza1}, where the authors showed how to perform a $S$-expansion from $osp(32/1)$ to a D'Auria-Fr\'{e}-like superalgebra (with the same structure of the D'Auria-Fr\'{e} superalgebra, but with different details), using $S^{(3)}_E$ as semigroup. This analogy confirms and corroborates the analytic method developed in the present work.

In this example, we also analyze the link between $osp(32/1)$ and another superalgebra included in the dichotomic solution found in \cite{D'Auria, Ravera}, in which the translations and the fermionic generators, respectively denoted by $\tilde{P}_a$ and $\tilde{Q}$, commute. We will see that the supersymmetric Lie algebra $osp(32/1)$ and this particular hidden superalgebra are linked by a $S$-expansion ($0_S$-resonant-reduction) procedure, in which
the semigroup involved in the process is the semigroup $S^{(2)}_E$, which satisfies the multiplication rules
\begin{equation}
\lambda_\alpha \lambda_\beta = \left\{ \begin{aligned} &
\lambda_{\alpha+\beta} , \;\;\;\;\; \text{when} \; \alpha+\beta \leqslant 3,
\\ & \lambda_3 , \;\;\;\;\;\;\;\;\; \text{when} \; \alpha + \beta > 3. \end{aligned} 
\right.
\end{equation}

We now want to find the correct semigroup leading from $osp(32/1)$ to the hidden superalgebra underlying $D=11$ supergravity through our analytic method. 
Let us start from collecting the useful information coming from the starting algebra $osp(32/1)$.
The generators of $osp(32/1)$ are, with respect to the Lorentz subgroup $SO(1,10)\subset osp(32/1)$, the following set of tensors (or spinors)
\begin{equation}
\lbrace{P_a,J_{ab},Z_{a_1...a_5},Q_\alpha \rbrace},
\end{equation}
where $J_{ab}$, $P_a$, $Q_\alpha$ can be respectively interpreted as the Lorentz, translations and supersymmetry generators, and where $Z_{a_1...a_5}$ is a $5$-index skew-symmetric generator associated with the physical $A_{\mu \nu \rho}$ field appearing in $D=11$ supergravity.

Now we have to take into account the information coming from the target superalgebra, that is to say the hidden superalgebra underlying the eleven-dimensional supergravity \cite{D'Auria, Ravera}. The generators of the mentioned superalgebra are given by the set
\begin{equation}
\lbrace \tilde{P}_a, \tilde{J}_{ab}, \tilde{Z}_{ab},\tilde{Z}_{a_1...a_5},\tilde{Q}_\alpha, \tilde{Q}'_\alpha \rbrace ,
\end{equation}
where $\tilde{Z}_{ab}, \tilde{Z}_{a_1 \cdots a_5}$ are two extra bosonic generators, and where $\tilde{Q}'$ is an extra fermionic generator that controls the gauge symmetry of the theory and allows the closure of the algebra.

We perform the detailed calculations in Appendix \ref{osp}, while in the following we summarize our results.

At the end of the whole procedure, we are left with the following multiplication table:
\begin{equation}
\begin{array}{c|ccccc}
 & \lambda _{0} & \lambda _{1} & \lambda _{2} & \lambda _{3} & \lambda_{4} \\ 
\hline
\lambda _{0} & \lambda _{0} & \lambda _{1} & \lambda _{2} & \lambda _{3} & \lambda_{4}
\\ 
\lambda _{1} & \lambda _{1} & \lambda _{2} & \lambda _{3} & \lambda _{4} & \lambda_{4}
\\ 
\lambda _{2} & \lambda _{2} & \lambda _{3} & \lambda _{4} & \lambda _{4} & \lambda_{4}
\\ 
\lambda _{3} & \lambda _{3} & \lambda _{4} & \lambda _{4} & \lambda _{4} & \lambda_{4}
\\
\lambda _{4} & \lambda _{4} & \lambda _{4} & \lambda _{4} & \lambda _{4} & \lambda_{4}
\end{array}
\end{equation}
which is the multiplication table describing the semigroup $S^{(3)}_E$, that, as it was also shown in \cite{Iza1}, is exactly the semigroup leading, through a $S$-expansion procedure ($0_S$-resonant-reduction), from the $osp(32/1)$ algebra to the hidden superalgebra described in \cite{D'Auria, Ravera}. Thus, we have shown that our analytic method immediately allows us to discover that these two superalgebras can be linked through a $S$-expansion procedure ($0_S$-resonant-reduction), involving the semigroup $S^{(3)}_E$. 

We now make some consideration on the case in which 
\begin{equation}
\left[\tilde{Q}, \tilde{P}_a\right]= 0 ,
\end{equation}
that is one of the commutation relations the other superalgebra presented in \cite{D'Auria}. In this case, from the relation
\begin{equation}
\left[\tilde{P}_a, \tilde{Q} \right] = \left[\lambda_e P_a, \lambda_c Q \right]=\lambda_e \lambda_c \left[P_a,Q \right]= 0 \;\;\; \rightarrow \;\;\; \lambda_e \lambda_c = \lambda_{0_S} , 
\end{equation}
we observe that we have to fix
\begin{equation}
\lambda_b = \lambda_e ,
\end{equation}
as we have done (see Appendix \ref{osp}) in the previous case, and also
\begin{equation}
\lambda_d = \lambda_{0_S},
\end{equation}
in order to have consistent multiplication rules. Thus, following the usual procedure, we can build the multiplication table of the set $\tilde{S}$, which in this case reads
\begin{equation}
\begin{array}{c|cccc}
 & \lambda _{0} & \lambda _{1} & \lambda _{2} & \lambda _{3}  \\ 
\hline
\lambda _{0} & \lambda _{0} & \lambda _{1} & \lambda _{2} & \lambda _{3} 
\\ 
\lambda _{1} & \lambda _{1} & \lambda _{2} & \lambda _{3} & \lambda _{3}
\\ 
\lambda _{2} & \lambda _{2} & \lambda _{3} & \lambda _{3} & \lambda _{3}
\\ 
\lambda _{3} & \lambda _{3} & \lambda _{3} & \lambda _{3} & \lambda _{3}
\end{array}
\end{equation}
This is exactly the multiplication table describing the semigroup $S^{(2)}_E$, which satisfies the multiplication rules 
\begin{equation}
\lambda_\alpha \lambda_\beta = \left\{ \begin{aligned} &
\lambda_{\alpha+\beta} , \;\;\;\;\; \text{when} \; \alpha+\beta \leqslant 3,
\\ & \lambda_3 , \;\;\;\;\;\;\;\;\; \text{when} \; \alpha + \beta > 3. \end{aligned} 
\right.
\end{equation}

In Ref. \cite{Iza1}, the authors showed that $S^{(2)}_E$ is the semigroup allowing the $S$-expansion from $osp(32/1)$ to the $M$-algebra (the algebra of the $M$-theory). We have now shown that the same result is reproduced when we are dealing with a $S$-expansion ($0_S$-resonant-reduction) from $osp(32/1)$ to a particular subalgebra of the hidden superalgebra obtained in \cite{D'Auria, Ravera} (that is to say, in the case in which $\tilde{Q}$ and $\tilde{P}_a$ commute). However, this particular subalgebra can be obtained with $S$-expansion from $osp(32/1)$ only if it coincides with the $M$-algebra. In fact, in this case the extra-fermionic generator of the target hidden superalgebra goes to zero.

Thus, a strong relation between the D'Auria-Fr\'{e} superalgebra and the $M$-algebra is evident. 
Both of them, as it was shown in \cite{Iza1}, can be reached with a $S$-expansion from $osp(32/1)$, respectively with the semigroup $S^{(2)}_E$ and $S^{(3)}_E$ (and this fact furnished us another corroboration of the analytic method developed in Section \ref{Method}). 

Furthermore, in our work we have interestingly shown that the particular subalgebra (where $\tilde{P}_a$ and $\tilde{Q}$ commute) of the hidden superalgebra underlying $D=11$ supergravity is linked to $osp(32/1)$ by the semigroup $S^{(2)}_E$, only if this coincides with the $M$-algebra described in \cite{Iza1}.

Finally, we conclude our observations saying that, previously in \cite{Bandos} and later in \cite{Ravera}, the authors also found a singular solution, which, in our notation, corresponds to consider the singular limit $\tilde{Q}' \rightarrow 0$ in the target superalgebra. 
In this case, the authomorphism group of the FDA is enlarged to $Sp(32)$, and the whole procedure resembles an In\"on\"u-Wigner contraction. 
Thus, we can clearly see the existence of a strong link between the mentioned superalgebras, given by the $S$-expansion ($0_S$-resonant-reduction) and the In\"on\"u-Wigner contraction procedures. However, we will not treat the case involving the In\"on\"u-Wigner contraction in our work, and we leave these considerations for the future.   

\section{Comments and possible developments}\label{Comments}

As we have previously said in the Introduction, a fundamental task to accomplish when dealing with the $S$-expansion is to find the appropriate semigroup linking two different (super)algebras, but this procedure is not a trivial one, and usually requires a kind of ``\textit{trial and error}" process.
In this paper, we have presented an \textit{analytic method} able to give us the multiplication table(s) of the set(s) involved in an $S$-expansion process (with either resonance or $0_S$-resonant-reduction) for reaching a target Lie (super)algebra from a starting one, after having properly chosen the partitions over subspaces of the considered (super)algebras.
The analytic method described in this work gives a simple set of expressions to find the subset decomposition of the set(s) involved in the process. Then, one can use the information coming from both the initial (super)algebra and the target one, in order to write the multiplication table(s) of the set(s). At the end of the procedure, one can check associativity by hand or with a simple computational algorithm (as we have done in this work), and thus end up with the complete multiplication table(s) of the semigroup(s) involved in the process.
We have then given some interesting examples of application, starting from simple cases and ending with a particular case involving supersymmetric algebras.
With these examples, we have reproduced well known results, which have already been presented in the literature, and we have also generalized some of them. We can thus conclude that our analytic method is reliable, and it can also be used in more complicated cases.

Future work can include the study of the particular cases in which the number of subspaces partitions of the target (super)algebra is \textit{different} from the number of subspaces partition of the starting one, and the extension (and generalization) of our analytic method to the case of infinite algebras and semigroups.
Furthermore, our analytic procedure can be used in future works for understanding the possible links that could exist between different (super)algebras (also in higher dimensional cases) that have not yet been analyzed in the $S$-expansion context.

\section{Acknowledgment}

We have benefited of stimulating discussions with L. Andrianopoli, R. D'Auria, and M. Trigiante, and we thank them for a critical reading of the manuscript.
The authors also wish to thank R. Caroca, P.K. Concha, N. Merino, E.K. Rodr\'{\i}guez, and P. Salgado for the enlightening suggestions.

Two of the authors (M. C. Ipinza and D. M. Pe\~{n}afiel) were supported by grants from the
\textit{Comisi\'{o}n Nacional de Investigaci\'{o}n Cient\'{i}fica y Tecnol\'{o}gica} CONICYT and from the Universidad de Concepci\'{o}n, Chile.

\appendix

\section{Detailed calculations for moving from the Bianchi Type I algebra (BTI) to the Bianchi Type II algebra (BTII)}\label{Bianchi}

The BTI and BTII Lie algebras have two and three generators, respectively.
The only commutator different from zero for the BTI algebra is
\begin{align}
\left[X_1,X_2\right]=X_1,
\end{align}
where $X_1$ and $X_2$ are the generators of the BTI algebra. For the BTII algebra we have
\begin{align}
\left[Y_1,Y_2\right]=&0, \\
\left[Y_1,Y_3\right]=&0, \\
\left[Y_2,Y_3\right]=&Y_1,
\end{align}
where $Y_1$, $Y_2$, and $Y_3$ are the generators of the BTII algebra.

Let us consider the following subspaces partition for the BTI algebra:
\begin{align}
\left[V_0,V_0\right]\subset & V_0 , \\ 
\left[V_0,V_1\right]\subset & V_0 \oplus V_1 , \\ 
\left[V_1,V_1\right]\subset & V_0 , 
\end{align}
where we have set $V_0=\lbrace0\rbrace \cup \lbrace X_2\rbrace$, and $V_1=\lbrace X_1\rbrace$.
Similarly, we can write the subspaces partition for the target BTII algebra:
\begin{align}
\left[\tilde{V}_0,\tilde{V}_0\right]\subset & \tilde{V}_0 , \\ 
\left[\tilde{V}_0,\tilde{V}_1\right]\subset & \tilde{V}_0 \oplus \tilde{V}_1 , \\ 
\left[\tilde{V}_1,\tilde{V}_1\right]\subset & \tilde{V}_0 , 
\end{align}
where we have denoted with $\tilde{V}_A$, $A=0,1$, the subspaces related to the target algebra and where we have defined $\tilde{V}_0=\lbrace0\rbrace \cup \lbrace Y_3\rbrace$, and $\tilde{V}_1=\lbrace Y_1, \; Y_2\rbrace$. Let us observe that, in this way, we have the same partition structure both for the initial algebra and for the target one.

We now follow the steps described in the analytic procedure of Section \ref{Method}, in order to obtain the possible abelian set(s) (with respect to the chosen partitions) leading from the BTI algebra to the BTII one.

First of all, we solve the system (\ref{system}), that in this case reads
\begin{equation}
\left\{
\begin{aligned}
1=& \;1\cdot \left(\Delta_0\right), \\
2=& \; 1\cdot \left(\Delta_1\right), \\
\tilde{P}=& \; \Delta_0 + \Delta_1 +1, 
\end{aligned} \right.
\end{equation}
since 
\begin{align}
dim(\tilde{V}_0)=1, \;\;\; dim(V_0)=1, \\ 	
dim(\tilde{V}_1)=2, \;\;\; dim(V_1)=1,
\end{align}
neglecting the zero element of the subspaces $V_0$ and $\tilde{V}_0$. Here we have denoted, as usual, with $\Delta_A$, $A=0,1$, the cardinality of the subsets $S_{\Delta_A}$ associated with the subspace $A$, \textit{i.e.} the number of elements in $S_{\Delta_A}$. Solving the system above, we get the unique solution
\begin{equation}
\tilde{P}=4, \;\;\; \Delta_0 = 1, \;\;\; \Delta_1 =2.
\end{equation}

Now, since $\tilde{S}=\lbrace S_{1_0}\rbrace \sqcup \lbrace S_{2_1} \rbrace \cup \lbrace \lambda_{0_S}\rbrace$, where $\lambda_{0_S}$ is the zero element of the set $\tilde{S}$, we can write the following subset decomposition structure of the set $\tilde{S}$:
\begin{align}\label{sgrsubset}
S_{1_0} = & \lbrace \lambda_a\rbrace  , \nonumber \\
S_{2_1} = & \lbrace \lambda_b , \; \lambda_c \rbrace .
\end{align}
Here the index $a,b,c$ identify general elements of the set $\tilde{S}$, and they are not running index. We do not yet identify them with numbers, because there still exists the possibility of having the same element in different subsets. 

The next step consists in finding the multiplication rules between the elements of each subset in (\ref{sgrsubset}). Thus we write the adjoint representation of the BTI algebra with respect to the subspaces partition:
\begin{equation}
(C)_{0B}^{C}=\begin{pmatrix}
(C)_{00}^{0} & 0 \\ 
0 & (C)_{01}^{1}
\end{pmatrix} , \;\;\; (C)_{1B}^{C}=\begin{pmatrix}
0 & (C)_{10}^{1} \\ 
(C)_{11}^{0} & 0
\end{pmatrix},
\end{equation}
where the index $B,C$ can assume the values $0$ or $1$, labeling the different subspaces involved in the partition of the algebra. Thus, we are now able to write the relations (\ref{eqsemigrold}), which, in this case, read
\begin{equation}
\begin{aligned}
& \left[ \left(  S_{ 1_{0}} \otimes V_{0}\right) \oplus \left(\lbrace \lambda_{0_S}\rbrace \otimes V_{0} \right) ,\left(S_{1 _{0}} \otimes V_{0}\right) \oplus \left(\lbrace \lambda_{0_S}\rbrace \otimes V_{0} \right)\right] = \\
& = \left( K_{\left(  1_{0}\right) \left(  1_{0}\right) }^{\left(  1_{0}\right) }\left( C\right)_{00}^{0}\right) \left(  S _{ 1 _{0}}\otimes V_{0}\right)\oplus \left(\lbrace \lambda_{0_S}\rbrace \otimes V_{0} \right), \\
& \left[ \left(  S_{2 _{1}} \otimes V_{1}\right)\oplus \left(\lbrace \lambda_{0_S}\rbrace \otimes V_{1} \right) ,\left( S_{1 _{0}} \otimes V_{0}\right)\oplus \left(\lbrace \lambda_{0_S}\rbrace \otimes V_{0} \right) \right] = \\
& = \left( K_{\left(  2 _{1}\right) \left(  1_{0}\right) }^{\left(  2_{1}\right) }\left( C\right)_{10}^{1}\right) \left( S _{ 2 _{1}}\otimes V_{1}\right)\oplus \left(\lbrace \lambda_{0_S}\rbrace \otimes V_{1} \right), \\
& \left[ \left(  S_{2_{1}}\otimes V_{1}\right) \oplus \left(\lbrace \lambda_{0_S}\rbrace \otimes V_{1} \right),\left( S_{2_{1}} \otimes V_{1}\right) \oplus \left(\lbrace \lambda_{0_S}\rbrace \otimes V_{1} \right)\right] = \\
& = \left( K_{\left(  2_{1}\right) \left(  2_{1}\right) }^{\left(   1_{0}\right) }\left( C\right)_{11}^{0}\right) \left(  S _{ 1 _{0}} \otimes V_{0}\right)\oplus \left(\lbrace \lambda_{0_S}\rbrace \otimes V_{0} \right).
\end{aligned}
\end{equation}
These relations (which can also be rewritten in a simpler form, such as the one in (\ref{eqsemigr})) give us a first view on the possibilities allowed by the multiplication table of the set $\tilde{S}$. In fact, we can now write
\begin{align}
S_{1_0}  \cdot  S_{1_0}  \subset & S_{1_0} \cup \lbrace \lambda_{0_S}\rbrace,\\
S_{2_1} \cdot  S_{1_0}  \subset & S_{2_1} \cup \lbrace \lambda_{0_S}\rbrace,\\
S_{2_1} \cdot S_{2_1}  \subset & S_{1_0} \cup \lbrace \lambda_{0_S}\rbrace,
\end{align}
where we have taken into account the presence of the zero element $\lambda_{0_S}$ of the set $\tilde{S}$. Thus, we are now able to write the possible multiplication rules between the elements of the set $\tilde{S}$, namely
\begin{align}
\lambda_{a}  \lambda_{a}= &\lambda_{a,0_S},\\
\lambda_{b,c} \lambda_{a}= &\lambda_{b,c,0_S},\\
\lambda_{b,c} \lambda_{b,c}=& \lambda_{a,0_S},
\end{align}
where we have already taken into account the triviality of the multiplications rules
\begin{align}
\lambda_{0_S}\lambda_{0_S}= &\lambda_{0_S}, \\
\lambda_{0_S}\lambda_{a,b,c}= &\lambda_{0_S}. 
\end{align}

We have exhausted the information coming from the initial algebra, thus we now use the information coming from the target one, in order to build up the complete multiplication table of the set $\tilde{S}$. 

We proceed by writing the relations between the $S$-expanded generators of the initial BTI algebra and the generators of the target BTII one, according to the usual $S$-expansion procedure described \cite{Iza1}. According to the identification criterion presented in Subsection \ref{IdCr}, we can perform the identification 
\begin{align}
\lambda_a X_2=&Y_3 \\
\lambda_b X_1=&Y_1\\
\lambda_c X_1=&Y_2.
\end{align}
We now write the commutators of the target BTII algebra in terms of the commutators between the $S$-expanded, resonant-reduced generators of the BTI one:
\begin{align}
\left[Y_2,Y_3\right]=&Y_1, \nonumber \\
\left[\lambda_c X_1,\lambda_a X_2\right]=&\lambda_b X_1, \nonumber\\
\lambda_c\lambda_a \left[X_1,X_2\right]=&\lambda_b X_1 \label{mult}.
\end{align}
Since for the BTI algebra we have $\left[X_1,X_2\right]=X_1$, from equation (\ref{mult}) we obtain
\begin{equation}
\lambda_c \lambda_a =\lambda_b .
\end{equation}
This simple analysis can be performed in order to find the correct multiplication rules between the elements of the set $\tilde{S}$, thus we proceed in this way, computing the other commutators and factorizing the product between the elements of $\tilde{S}$, in order to end up with the complete multiplication table.

Let us observe that the commutator
\begin{equation}
\left[ Y_1, Y_2 \right]=0
\end{equation}
does not give us any further information about the multiplication rule between $\lambda_b$ and $\lambda_c$, due to the fact that, when we write it in terms of the commutator between $S$-expanded generators and we factorize the product between the elements of $\tilde{S}$, we are left with $\lambda_b \lambda_c \left[X_1,X_1\right]=0$, which reproduces a trivial identity, since $\left[X_1,X_1\right]=0$ in the BTI algebra.
On the other hand, from the study of the last commutator we have to consider, we get
\begin{align}
\left[Y_1,Y_3\right]=&0, \nonumber \\
\left[\lambda_bX_1,\lambda_aX_2\right]=&0, \nonumber\\
\lambda_b\lambda_a \left[X_1,X_2\right]=&0 \label{mult2}.
\end{align}
Since, from the initial BTI algebra, we know that $\left[X_1,X_2\right]=X_1\neq 0$, from equation (\ref{mult2}) we clearly see that the the zero element $\lambda_{0_S}$ is naturally involved in the procedure, and we can finally write
\begin{equation}
\lambda_b \lambda_a = \lambda_{0_S} .
\end{equation}

Summarizing, by following the identification criterion described in Subsection \ref{IdCr}, we have obtained the multiplication rules
\begin{align}
\lambda_c  \lambda_a =& \lambda_b, \\
\lambda_b  \lambda_a =& \lambda_{0_S} .
\end{align}
Let us notice that these multiplication rules are consistent with those previously obtained along the procedure, when we have exploited the information coming from the initial BTI algebra. 

We are now able to write the following multiplication tables for the sets $\tilde{S}$'s involved in the procedure:
\begin{equation}
\begin{array}{c|cccc}
 &\lambda_a & \lambda_b & \lambda_c &\lambda_{0_S} \\
\hline
\lambda_a & \lambda_{a,0_S} & \lambda_{0_S}& \lambda_{b}&\lambda_{0_S}\\
\lambda_b & \lambda_{0_S}& \lambda_{a,0_S} & \lambda_{a,0_S}&\lambda_{0_S} \\
\lambda_c & \lambda_b & \lambda_{a,0_S}& \lambda_{a,0_S}&\lambda_{0_S}\\
\lambda_{0_S} &\lambda_{0_S} &\lambda_{0_S} &\lambda_{0_S} &\lambda_{0_S}
\end{array}
\end{equation}
These are the multiplication tables of the possible sets $\tilde{S}$'s (with respect to the chosen partitions) involved in the $S$-expansion, $0_S$-resonant-reduced procedure for moving from the BTI algebra to the BTII one.

\section{Detailed calculations for reaching $iso(2,1)$, starting from $so(2,2)$}\label{soiso}

Both $iso(2,1)$ and $so(2,2)$ have six generators. 
The commutation relations between the generators of the starting $so(2,2)$ algebra are
\begin{align}
& \left[ J_{i},J_{j}\right] = \epsilon_{ijk} J^{k}, \\
& \left[ J_{i},P_{j}\right] = \epsilon_{ijk} P^{k}, \\
& \left[ P_{i},P_{j}\right] = \epsilon_{ijk} J^{k},
\end{align}
where $i,j,k,...=1,2,3$, and the commutation relations between the
generators of the target algebra $iso(2,1)$ read 
\begin{align}
& \left[ \tilde{J}_{i},\tilde{J}_{j}\right] = \epsilon_{ijk} \tilde{J}^{k}, \\
& \left[ \tilde{J}_{i},\tilde{P}_{j}\right] = \epsilon_{ijk} \tilde{P}^{k}, \\
& \left[ \tilde{P}_{i},\tilde{P}_{j}\right] = 0,
\end{align}%
where, again, $i,j,k,...=1,2,3$. 

We can write the following subspaces partition: 
\begin{align}
\left[ V_{0},V_{0}\right] \subset & V_{0}, \\
\left[ V_{0},V_{1}\right] \subset & V_{1}, \\
\left[ V_{1},V_{1}\right] \subset & V_{0}
\end{align}%
for $so(2,2)$, where we have set $V_{0}=\{J_{i}\}$ and $V_{1}=\{P_{i}\}$, and 
\begin{align}
\left[ \tilde{V}_{0},\tilde{V}_{0}\right] \subset & \tilde{V}_{0}, \\
\left[ \tilde{V}_{0},\tilde{V}_{1}\right] \subset & \tilde{V}_{1}, \\
\left[ \tilde{V}_{1},\tilde{V}_{1}\right] \subset & \tilde{V}_{0}
\end{align}%
for $iso(2,1)$, where, in analogy to what we have done for the initial
algebra, $V_{0}=\lbrace 0 \rbrace \oplus \{\tilde{J}_{i}\}$ and $V_{1}=\{\tilde{P}_{i}\}$. Thus, we can now write: 
\begin{align}
dim(V_{0})=& 3, \\
dim(V_{1})=& 3
\end{align}%
for the $so(2,2)$ algebra, and, similarly, 
\begin{align}
dim(\tilde{V}_{0})=& 3, \\
dim(\tilde{V}_{1})=& 3
\end{align}%
for the $iso(2,2)$ one. Now we have all the information we need to know for
proceeding. Thus, we move to the study of the system (\ref{system}), which, in this case, reads 
\begin{equation}
\left\{
\begin{aligned}
3=& \;3\cdot \left( \Delta_0 \right), \\
3=& \;3\cdot \left( \Delta_1 \right), \\
\tilde{P}=& \; \Delta _{1}+\Delta _{0}+1.
\end{aligned} \right.
\end{equation}
This system admits the unique solution 
\begin{equation}
\tilde{P}=3,\;\;\;\Delta _{0}=1,\;\;\;\Delta _{1}=1.
\end{equation}
Thus, we now know that the set $\tilde{S}=\{{S_{1_0}}\}\sqcup \{{S_{1_1}}\} \cup \lbrace \lambda_{0_S}\rbrace$, involved in the $0_{S}$-resonant-reduction process to reach the algebra $iso(2,1)$ starting from $so(2,2)$, must have the following subset decomposition structure: 
\begin{align}
S_{1_{0}}=& \lbrace{\lambda _{a}\rbrace}, \\
S_{1_{1}}=& \lbrace{\lambda _{b}\rbrace}.
\end{align}

We can now write down the adjoint
representation of $so(2,2)$ with respect to the subspaces partition: 
\begin{equation}
(C)_{0B}^{C}=
\begin{pmatrix}
(C)_{00}^{0} & 0 \\ 
0 & (C)_{01}^{1}
\end{pmatrix}
,\;\;\;(C)_{1B}^{C}=
\begin{pmatrix}
0 & (C)_{10}^{1} \\ 
(C)_{11}^{0} & 0
\end{pmatrix},
\end{equation}
where the index $B,C$ can assume the values $0$ or $1$, labeling the different subspaces,
and, subsequently, we are now able to write relations of the type (\ref{eqsemigrold}), which, in this case, read
\begin{equation}
\begin{aligned}
& \left[ \left(  S_{1_{0}} \otimes V_{0}\right)\oplus \left(\lbrace \lambda_{0_S}\rbrace \otimes V_{0}\right) ,\left( S_{1_{0}} \otimes V_{0}\right) \oplus \left(\lbrace \lambda_{0_S}\rbrace \otimes V_{0}\right)\right]  = \\
& = \left( K_{\left(1_{0}\right) \left( 1_{0}\right) }^{\left( 1_{0}\right) }\left( C\right)
_{00}^{0}\right) \left( S_{1_{0}} \otimes V_{0}\right)\oplus \left(\lbrace \lambda_{0_S}\rbrace \otimes V_{0}\right) , \\
& \left[ \left( S_{1_{1}} \otimes V_{1}\right) \oplus \left(\lbrace \lambda_{0_S}\rbrace \otimes V_{1}\right),\left( S_{1_{0}} \otimes V_{0}\right) \oplus \left(\lbrace \lambda_{0_S}\rbrace \otimes V_{0}\right)\right]  = \\
& = \left( K_{\left(1_{1}\right) \left( 1_{0}\right) }^{\left( 1_{1}\right) }\left( C\right)
_{10}^{1}\right) \left(  S_{1_{1}} \otimes V_{1}\right)\oplus \left(\lbrace \lambda_{0_S}\rbrace \otimes V_{1}\right) , \\
& \left[ \left(  S_{1_{1}} \otimes V_{1}\right) \oplus \left(\lbrace \lambda_{0_S}\rbrace \otimes V_{1}\right),\left( S_{1_{1}}  \otimes V_{1}\right) \oplus \left(\lbrace \lambda_{0_S}\rbrace \otimes V_{1}\right)\right] = \\
& = \left( K_{\left(
1_{1}\right) \left( 1_{1}\right) }^{\left( 1_{0}\right) }\left( C\right)
_{11}^{0}\right) \left(  S_{1_{0}} \otimes V_{0}\right) \oplus \left(\lbrace \lambda_{0_S}\rbrace \otimes V_{0}\right).
\end{aligned}
\end{equation}
These relations (which can also be rewritten in a simpler form of the type (\ref{eqsemigr})) lead us to a first view on the multiplication rules between the
subsets of the set $\tilde{S}$, namely  
\begin{align}
S_{1_{0}}\cdot  S_{1_{0}}\subset & S_{1_{0}}\cup \lbrace \lambda_{0_S} \rbrace, \\
S_{1_{1}}  \cdot S_{1_{0}} \subset & S_{1_{1}}\cup \lbrace \lambda_{0_S} \rbrace, \\
S_{1_{1}}\cdot  S_{1_{1}}\subset & S_{1_{0}}\cup \lbrace \lambda_{0_S} \rbrace,
\end{align}
where we have taken into account the presence of the zero element $\lambda_{0_S}$ of the set $\tilde{S}$.

In terms of the elements of the subsets, we can now write
\begin{align}
\lambda _{a}\lambda _{a}=& \lambda _{a,0_S}, \\
\lambda _{b}\lambda _{a}=& \lambda _{b,0_S}, \\
\lambda _{b}\lambda _{b}=& \lambda _{a,0_S},
\end{align}
where we have already taken into account the triviality of the multiplications rules
\begin{align}
\lambda_{0_S}\lambda_{0_S}= &\lambda_{0_S}, \\
\lambda_{0_S}\lambda_{a,b}= &\lambda_{0_S}. 
\end{align}

Then, by exploiting the information coming from the target algebra, we will be able to fix the degeneracy present in the above multiplication rules.

To this aim, let us perform the following associations between the $S$-expanded generators of the starting algebra and the generators of the target one (according to the identification criterion described in Subsection \ref{IdCr}): 
\begin{align}
\lambda _{a}J_{i}=& \tilde{J}_{i}, \\
\lambda _{b}P_{i}=& \tilde{P}_{i}.
\end{align}
Now we can write, according to the $S$-expansion procedure \cite{Iza1}, the commutators of
the target $iso(2,1)$ algebra in terms of the commutators between the $S$-expanded generators of the $so(2,2)$ one: 
\begin{align}
\left[ \tilde{J}_{i},\tilde{J}_{j}\right] \propto & \tilde{J}_{k},  \notag \\
\left[ \lambda _{a}J_{i},\lambda _{a}J_{j}\right] \propto & \lambda
_{a}J_{k},  \notag \\
\lambda _{a}\lambda _{a}\left[ J_{i},J_{j}\right] \propto & \lambda
_{a}J_{k}.  \label{multJ}
\end{align}
Since $\left[ J_{i},J_{j}\right] \propto J_{k}$, equation (\ref{multJ})
tells us that 
\begin{equation}
\lambda _{a}\lambda _{a}=\lambda _{a}.
\end{equation}
Similarly, since $\left[ J_{i},P_{j}\right] \propto P_{k}$, the analysis of 
\begin{align}
\left[ \tilde{J}_{i},\tilde{P}_{j}\right] \propto & \tilde{P}_{k},  \notag \\
\left[ \lambda _{a}J_{i},\lambda _{b}P_{j}\right] \propto & \lambda
_{b}P_{k},  \notag \\
\lambda _{a}\lambda _{b}\left[ J_{i},P_{j}\right] \propto & \lambda _{b}P_{k}
\label{multJP}
\end{align}
gives us  
\begin{equation}
\lambda _{a}\lambda _{b}=\lambda _{b}.
\end{equation}
Finally, from 
\begin{align}
\left[ \tilde{P}_{i},\tilde{P}_{j}\right] =0,  \notag \\
\left[ \lambda _{b}P_{i},\lambda _{b}P_{j}\right] =0,  \notag \\
\lambda _{b}\lambda _{b}\left[ P_{i},P_{j}\right] =0,  \label{multP}
\end{align}
we can write 
\begin{equation}
\lambda _{b}\lambda _{b}=\lambda _{0_S},
\end{equation}
since $\left[ P_{i},P_{j}\right] \propto J_{k} \neq 0$.

Summarizing, we are left with the multiplication rules 
\begin{align}
& \lambda _{a}\lambda _{a}=\lambda _{a}, \\
& \lambda _{a}\lambda _{b}=\lambda _{b}, \\
& \lambda _{b}\lambda _{b}=\lambda _{0_S}.
\end{align}
In this way, we have completely fixed the degeneracy appearing in the
multiplication rules between the elements of the set $\tilde{S}$ and we are finally able to write the multiplication table of the set $\tilde{S}$, which reads
\begin{equation}
\begin{array}{c|ccc}
& \lambda _{a} & \lambda _{b} & \lambda _{0_S} \\ \hline
\lambda _{a} & \lambda _{a} & \lambda _{b} & \lambda _{0_S} \\ 
\lambda _{b} & \lambda _{b} & \lambda _{0_S} & \lambda _{0_S} \\ 
\lambda _{0_S} & \lambda _{0_S} & \lambda _{0_S} & \lambda _{0_S}
\end{array}
\label{tabla123v2general}
\end{equation}
Here we have completed the multiplication table with the zero element $\lambda _{0_S}$, by exploiting the definition $\lambda _{0_S}\lambda _{a,b}=\lambda _{0_S}$.

After having performed the identification 
\begin{align}
& \lambda _{a}\leftrightarrow \lambda _{0}, \\
& \lambda _{b}\leftrightarrow \lambda _{1}, \\
& \lambda _{0_S}\leftrightarrow \lambda _{2},
\end{align}
we can write the above multiplication table as follow:
\begin{equation}
\begin{array}{c|ccc}
& \lambda _{0} & \lambda _{1} & \lambda _{2} \\ \hline
\lambda _{0} & \lambda _{0} & \lambda _{1} & \lambda _{2} \\ 
\lambda _{1} & \lambda _{1} & \lambda _{2} & \lambda _{2} \\ 
\lambda _{2} & \lambda _{2} & \lambda _{2} & \lambda _{2}
\end{array}
\label{tabla123v2}
\end{equation}
Table (\ref{tabla123v2}) is an abelian and associative one (the check for associativity can be performed either by hand or using a simple computational algorithm), and thus it describes the semigroup leading from the $so(2,2)$ algebra to the $iso(2,1)$ one (with respect to the partitions over subspaces that we have chosen).

\section{Detailed calculations for moving from the Anti-de Sitter ($AdS$) algebra to the Maxwell algebra}\label{Maxwell}

The $AdS$ Lie algebra has ten generators, while the Maxwell algebra ($\mathcal{M}$) counts sixteen generators. 
The generators of the $AdS$ algebra are $\lbrace J_{ab},P_a \rbrace$, and they satisfy the
commutation relations 
\begin{align}
& \left[ J_{ab}, J_{cd}\right] =
\eta_{bc}J_{ad}-\eta_{ac}J_{bd}-\eta_{bd}J_{ac}+\eta_{ad}J_{bc}, \\
& \left[J_{ab},P_c \right] = \eta_{bc}P_a - \eta_{ac}P_b , \\
& \left[P_a, P_b \right] = J_{ab}.
\end{align}

The generators of the Maxwell algebra $\mathcal{M}$ are $\{\tilde{J}_{ab},\tilde{P}_{a},\tilde{Z}_{ab}\}$, and they satisfy the following commutation relations: 
\begin{align}
& \left[ \tilde{J}_{ab},\tilde{Z}_{cd}\right] =\eta _{bc}\tilde{Z}_{ad}-\eta
_{ac}\tilde{Z}_{bd}-\eta _{bd}\tilde{Z}_{ac}+\eta _{ad}\tilde{Z}_{bc}, \\
& \left[ \tilde{Z}_{ab},\tilde{P}_{a}\right] =0, \\
& \left[ \tilde{Z}_{ab},\tilde{Z}_{cd}\right] =0, \\
& \left[ \tilde{J}_{ab},\tilde{J}_{cd}\right] =\eta _{bc}\tilde{J}_{ad}-\eta
_{ac}\tilde{J}_{bd}-\eta _{bd}\tilde{J}_{ac}+\eta _{ad}\tilde{J}_{bc}, \\
& \left[ \tilde{J}_{ab},\tilde{P}_{c}\right] =\eta _{bc}\tilde{P}_{a}-\eta
_{ac}\tilde{P}_{b}, \\
& \left[ \tilde{P}_{a},\tilde{P}_{b}\right] =\tilde{Z}_{ab}.
\end{align}
The subspace structure of the $AdS$ Lie algebra can be written as
\begin{align}
\left[V_0,V_0\right] \subset & V_0 , \\
\left[V_0,V_1\right] \subset & V_1, \\
\left[V_1,V_1 \right] \subset & V_0,
\end{align}
where $V_0= \lbrace{J_{ab} \rbrace}$ and $V_1=\lbrace P_a\rbrace$.

The subspace structure of the Maxwell algebra $\mathcal{M}$, in analogy to what
we have done for the $AdS$ algebra, may be written as 
\begin{align}
\left[ \tilde{V}_{0},\tilde{V}_{0}\right] \subset & \tilde{V}_{0}, \\
\left[ \tilde{V}_{0},\tilde{V}_{1}\right] \subset & \tilde{V}_{0}\oplus 
\tilde{V}_{1}, \\
\left[ \tilde{V}_{1},\tilde{V}_{1}\right] \subset & \tilde{V}_{0},
\end{align}
where $\tilde{V}_{0}=\{0\}\cup \{{\tilde{J}_{ab},\tilde{Z}_{ab}\}}$ and $%
\tilde{V}_{1}=\{\tilde{P}_{a}\}$. Thus, we have 
\begin{align}
dim(V_{0})=& 6, \\
dim(V_{1})=& 4
\end{align}%
for the $AdS$ algebra, and 
\begin{align}
dim(\tilde{V}_{0})=& 12, \\
dim(\tilde{V}_{1})=& 4
\end{align}%
for the Maxwell one. 

Now, we can solve the usual system (\ref{system}), which
in this case reads 
\begin{equation}
\left\{
\begin{aligned}
&12=  6\cdot \left(\Delta_0 \right) , \\
& \; 4= 4\cdot \left(\Delta_1 \right) , \\
& \tilde{P} =\Delta _{0}+\Delta _{1}+1.
\end{aligned} \right.
\end{equation}
This system has the unique solution 
\begin{equation}
\tilde{P}=4,\;\;\;\Delta _{0}=2,\;\;\;\Delta _{1}=1 .
\end{equation}
Thus, we now know that the set $\tilde{S}=\{{S_{2_{0}}}\}\sqcup \{{S_{1_{1}}}\}\cup\lbrace \lambda_{0_S}\rbrace$, involved in the $0_S$-resonant-reduction procedure to reach the Maxwell algebra $\mathcal{M}$ starting from the $AdS$ Lie one, must have the following subset decomposition structure: 
\begin{align}
{S_{2_{0}}}=& \lbrace{\lambda _{a},\lambda _{b}\rbrace}, \\
{S_{1_{1}}}=& \lbrace{\lambda _{c}\rbrace} .
\end{align}

By following the usual method (see Section \ref{Method}), now we write down the adjoint
representation of the $AdS$ algebra with respect to the subspaces
partition, namely 
\begin{equation}
(C)_{0B}^{C}=%
\begin{pmatrix}
(C)_{00}^{0} & 0 \\ 
0 & (C)_{01}^{1}%
\end{pmatrix}%
,\;\;\;(C)_{1B}^{C}=%
\begin{pmatrix}
0 & (C)_{10}^{1} \\ 
(C)_{11}^{0} & 0%
\end{pmatrix}%
,
\end{equation}%
where the index $B,C$ can take the values $0$ or $1$, labeling the different subspaces.

Consequently, we can write relations (\ref{eqsemigrold}) as follows: 
\begin{equation}
\begin{aligned}
& \left[ \left(  S_{2_{0}} \otimes V_{0}\right)\oplus \left(\lbrace \lambda_{0_S} \rbrace \otimes V_0 \right) ,\left( S_{2_{0}} \otimes V_{0}\right)\oplus \left(\lbrace \lambda_{0_S} \rbrace \otimes V_0 \right) \right] = \\
& = \left( K_{\left(2_{0}\right) \left( 2_{0}\right) }^{\left( 2_{0}\right) }\left( C\right)
_{00}^{0}\right) \left(  S_{2_{0}} \otimes V_{0}\right)\oplus \left(\lbrace \lambda_{0_S} \rbrace \otimes V_0 \right) , \\
& \left[ \left(  S_{1_{1}} \otimes V_{1}\right) \oplus \left(\lbrace \lambda_{0_S} \rbrace \otimes V_1 \right),\left( S_{2_{0}} \otimes V_{0}\right) \oplus \left(\lbrace \lambda_{0_S} \rbrace \otimes V_0 \right)\right] = \\
& = \left( K_{\left(1_{1}\right) \left( 2_{0}\right) }^{\left( 1_{1}\right) }\left( C\right)
_{10}^{1}\right) \left(  S_{1_{1}} \otimes V_{1}\right)\oplus \left(\lbrace \lambda_{0_S} \rbrace \otimes V_1 \right)  , \\
& \left[ \left( S_{1_{1}} \otimes V_{1}\right)\oplus \left(\lbrace \lambda_{0_S} \rbrace \otimes V_1 \right) ,\left( S_{1_{1}} \otimes V_{1}\right) \oplus \left(\lbrace \lambda_{0_S} \rbrace \otimes V_1 \right)\right] = \\
& = \left( K_{\left(
1_{1}\right) \left( 1_{1}\right) }^{\left( 2_{0}\right) }\left( C\right)
_{11}^{0}\right) \left(  S_{2_{0}} \otimes V_{0}\right) \oplus \left(\lbrace \lambda_{0_S} \rbrace \otimes V_0 \right) .
\end{aligned}
\end{equation}
These relations (which can be rewritten in the simpler form described in (\ref{eqsemigr})) tell us the composition rules of the subsets of the set $\tilde{S}$, which  read
\begin{align}
S_{2_{0}}\cdot S_{2_{0}}\subset & S_{2_{0}}\cup \lbrace \lambda_{0_S} \rbrace, \\
S_{2_{0}}\cdot S_{1_{1}}\subset & S_{1_{1}}\cup \lbrace \lambda_{0_S} \rbrace, \\
S_{1_{1}} \cdot  S_{1_{1}}\subset & S_{2_{0}}\cup \lbrace \lambda_{0_S} \rbrace ,
\end{align}
where we have explicitly taken into account the presence of the zero element $\lambda_{0_S}$.

Thus, we can write the following multiplication rules for the elements of the set $\tilde{S}$:
\begin{align}
\lambda _{a,b}\lambda _{a,b}=& \lambda _{a,b,0_S}, \\
\lambda _{a,b}\lambda _{c}=& \lambda _{c,0_S}, \\
\lambda _{c}\lambda _{c}=& \lambda _{a,b,0_S},
\end{align}
where we have already taken into account the triviality of the multiplications rules
\begin{align}
\lambda_{0_S}\lambda_{0_S}= &\lambda_{0_S}, \\
\lambda_{0_S}\lambda_{a,b,c}= &\lambda_{0_S}. 
\end{align}

Then, by exploiting the information coming from the target algebra,
we will be able to completely fix the degeneracy still appearing in the above multiplication
rules.

Thus, we write the generators of the Maxwell algebra $\mathcal{M}$ in terms
of the $S$-expanded generators of the $AdS$ Lie algebra, performing the
following identification (according to the identification criterion of Subsection \ref{IdCr}): 
\begin{align}
\lambda _{a}J_{ab}=& \tilde{J}_{ab}, \\
\lambda _{b}J_{ab}=& \tilde{Z}_{ab}, \\
\lambda _{c}P_{a}=& \tilde{P}_{a}.
\end{align}%
Consequently, according to the $S$-expansion procedure described in \cite{Iza1}, we can write the commutators of
the target Maxwell algebra in terms of the commutators between the $S$-expanded generators of the $AdS$ one, taking into account the commutation relations of the initial $AdS$ Lie algebra (in the following, we will
neglect, for simplicity, the index labeling the generators, since we just
need to exploit the commutators structure). 

We get
\begin{align}
\left[ \tilde{J},\tilde{J}\right] \propto & \tilde{J},  \notag \\
\left[ \lambda _{a}J,\lambda _{a}J\right] \propto & \lambda _{a}J,  \notag \\
\lambda _{a}\lambda _{a}\left[ J,J\right] \propto & \lambda
_{a}J\;\;\;\Rightarrow \;\;\;\lambda _{a}\lambda _{a}=\lambda _{a};
\end{align}
\begin{align}
\left[ \tilde{J},\tilde{P}\right] \propto & \tilde{P},  \notag \\
\left[ \lambda _{a}J,\lambda _{c}P\right] \propto & \lambda _{c}P,  \notag \\
\lambda _{a}\lambda _{c}\left[ J,P\right] \propto & \lambda
_{c}P\;\;\;\Rightarrow \;\;\;\lambda _{a}\lambda _{c}=\lambda _{c};
\end{align}
\begin{align}
\left[ \tilde{J},\tilde{Z}\right] \propto & \tilde{Z},  \notag \\
\left[ \lambda _{a}J,\lambda _{b}J\right] \propto & \lambda _{b}J,  \notag \\
\lambda _{a}\lambda _{b}\left[ J,J\right] \propto & \lambda
_{b}J\;\;\;\Rightarrow \;\;\;\lambda _{a}\lambda _{b}=\lambda _{b};
\end{align}
\begin{align}
\left[ \tilde{P},\tilde{P}\right] \propto & \tilde{Z},  \notag \\
\left[ \lambda _{c}P,\lambda _{c}P\right] \propto & \lambda _{b}J,  \notag \\
\lambda _{c}\lambda _{c}\left[ P,P\right] \propto & \lambda
_{b}J\;\;\;\Rightarrow \;\;\;\lambda _{c}\lambda _{c}=\lambda _{b};
\end{align}
\begin{align}
\left[ \tilde{Z},\tilde{P}\right] =& 0,  \notag \\
\left[ \lambda _{b}J,\lambda _{c}P\right] =& 0,  \notag \\
\lambda _{b}\lambda _{c}\left[ J,P\right] =& 0\;\;\;\Rightarrow
\;\;\;\lambda _{b}\lambda _{c}=\lambda _{0_S},
\end{align}
since $\left[ J,P\right] =P\neq 0$; 
\begin{align}
\left[ \tilde{Z},\tilde{Z}\right] =& 0,  \notag \\
\left[ \lambda _{b}J,\lambda _{b}J\right] =& 0,  \notag \\
\lambda _{b}\lambda _{b}\left[ J,J\right] =& 0\;\;\;\Rightarrow
\;\;\;\lambda _{b}\lambda _{b}=\lambda _{0_S},
\end{align}
since $\left[ J,J\right] =J\neq 0$. From the above relations, we can write the multiplication table
\begin{equation}
\begin{array}{c|cccc}
& \lambda _{a} & \lambda _{b} & \lambda _{c} & \lambda _{0_S} \\ 
\hline
\lambda _{a} & \lambda _{a} & \lambda _{b} & \lambda _{c} & \lambda _{0_S} \\ 
\lambda _{b} & \lambda _{b} & \lambda _{0_S} & \lambda _{0_S} & \lambda _{0_S} \\ 
\lambda _{c} & \lambda _{c} & \lambda _{0_S} & \lambda _{b} & \lambda _{0_S} \\ 
\lambda _{0_S} & \lambda _{0_S} & \lambda _{0_S} & \lambda _{0_S} & \lambda _{0_S}
\end{array}
\label{tablaMaxwellv2general}
\end{equation}
in which we can see that all the degeneracy has been fixed.
Then, after having performed the identification 
\begin{equation}
\lambda _{a}\leftrightarrow  \lambda _{0}, \;\;\; \lambda _{b}\leftrightarrow  \lambda _{2}, \;\;\;
\lambda _{c}\leftrightarrow  \lambda _{1}, \;\;\; \lambda _{0_S}\leftrightarrow  \lambda _{2},
\end{equation}
we can write the following multiplication table (where the elements of the set $\tilde{S}$ are written in the usual order): 
\begin{equation}
\begin{array}{c|cccc}
& \lambda _{0} & \lambda _{1} & \lambda _{2} & \lambda _{3} \\ 
\hline
\lambda _{0} & \lambda _{0} & \lambda _{1} & \lambda _{2} & \lambda _{3} \\ 
\lambda _{1} & \lambda _{1} & \lambda _{2} & \lambda _{3} & \lambda _{3} \\ 
\lambda _{2} & \lambda _{2} & \lambda _{3} & \lambda _{3} & \lambda _{3} \\ 
\lambda _{3} & \lambda _{3} & \lambda _{3} & \lambda _{3} & \lambda _{3}
\end{array}
\label{tablaMaxwellv2}
\end{equation}
Table (\ref{tablaMaxwellv2}) represents an abelian, commutative and associative semigroup, named $S^{(2)}_E$. As said in \cite{Concha1, Salgado, Concha2}, $S^{(2)}_E$ is
the semigroup involved in the $S$-expansion ($0_{S}$-resonant-reduction)
procedure performed in order to reach the Maxwell algebra $\mathcal{M}$
starting from the $AdS$ Lie algebra, and we have reproduced this result with our analytic method.

\section{Detailed calculations for reaching the hidden superalgebra underlying $D=11$ supergravity, starting from the supersymmetric Lie algebra $osp(32/1)$}\label{osp}

In the following, we display the detailed calculations for moving from the supersymmetric Lie algebra $osp(32/1)$ to the hidden superalgebra underlying $D=11$ supergravity, through a $0_S$-resonant-reduction procedure, and we show how to find the set(s) involved in the process, once the partitions over subspaces for both the considered superalgebras have been properly chosen.

The generators of $osp(32/1)$ are
\begin{equation}
\lbrace{P_a,J_{ab},Z_{a_1...a_5},Q_\alpha \rbrace}.
\end{equation}
The commutations relation between these generators can be written as  
\begin{eqnarray}
\left[ P_{a},P_{b}\right] &=&J_{ab} , \notag \\
\left[ J^{ab},P_{c}\right] &=&\delta _{ec}^{ab}P^{e},  \notag \\
\left[ J^{ab},J_{cd}\right] &=&\delta _{ecd}^{abf}\; J_{\; f}^{e} , \notag \\
\left[ P_{a},Z_{b_{1}\cdots b_{5}}\right] &=&-\frac{1}{5!}\epsilon
_{ab_{1}\cdots b_{5}c_{1}\cdots c_{5}}Z^{c_{1}\cdots c_{5}} , \notag \\
\left[ J^{ab},Z_{c_{1}\cdots c_{5}}\right] &=&\frac{1}{4!}\delta
_{dc_{1}\cdots c_{5}}^{abe_{1}\cdots e_{4}}Z_{\; e_{1}\cdots e_{4}}^{d},  \notag
\\
\left[ Z^{a_{1}\cdots a_{5}},Z_{b_{1}\cdots b_{5}}\right] &=&\eta ^{\left[
a_{1}\cdots a_{5}\right] \left[ c_{1}\cdots c_{5}\right] }\epsilon
_{c_{1}\cdots c_{5}b_{1}\cdots b_{5}e}P^{e}+\delta _{db_{1}\cdots
b_{5}}^{a_{1}\cdots a_{5}e}J_{\; e}^{d} + \notag \\
&&-\frac{1}{3!3!5!}\epsilon _{c_{1}\cdots c_{11}}\delta
_{d_{1}d_{2}d_{3}b_{1}\cdots b_{5}}^{a_{1}\cdots a_{5}c_{4}c_{5}c_{6}}\eta ^{%
\left[ c_{1}c_{2}c_{3}\right] \left[ d_{1}d_{2}d_{3}\right] }Z^{c_{7}\cdots
c_{11}} , \notag \\
\left[ P_{a},Q\right] &=&-\frac{1}{2}\Gamma _{a}Q , \notag \\
\left[ J_{ab},Q\right] &=&-\frac{1}{2}\Gamma _{ab}Q  ,\notag \\
\left[ Z_{abcde},Q\right] &=&-\frac{1}{2}\Gamma _{abcde}Q , \notag \\
\left\{ Q^{\rho },Q^{\sigma }\right\} &=&-\frac{1}{2^{3}}\left[ \left(
\Gamma ^{a}C^{-1}\right) ^{\rho \sigma }P_{a}-\frac{1}{2}\left( \Gamma
^{ab}C^{-1}\right) ^{\rho \sigma }J_{ab}\right] + \notag \\
&&-\frac{1}{2^{3}}\left[\frac{1}{5!}\left( \Gamma ^{abcde}C^{-1}\right) ^{\rho \sigma }Z_{abcde}\right] ,
\end{eqnarray}
where $C_{\rho \sigma}$ is the charge conjugation matrix and $\Gamma_a$, $\Gamma_{ab}$, $\Gamma_{abcde}$ are the Dirac matrices in eleven dimensions.
Let us perform the following subspaces partition for the $osp(32/1)$ algebra:
\begin{eqnarray}
\left[ V_{0},V_{0}\right] &\subset &V_{0} ,\\
\left[ V_{0},V_{1}\right] &\subset &V_{1} ,\\
\left[ V_{0},V_{2}\right] &\subset &V_{2} ,\\
\left[ V_{1},V_{1}\right] &\subset &V_{0}\oplus V_{2}, \\
\left[ V_{1},V_{2}\right] &\subset &V_{1} ,\\
\left[ V_{2},V_{2}\right] &\subset &V_{0}\oplus V_{2},
\end{eqnarray}
where we have set $V_0= \lbrace J_{ab} \rbrace$, $V_1=\lbrace Q_\alpha \rbrace$, and $V_2=\lbrace P_a, Z_{a_1...a_5} \rbrace$. Thus, the dimensions of the internal decomposition of $osp(32/1)$ read
\begin{eqnarray}
\dim \left( V_{0}\right) &=&\underset{J_{ab}}{\underbrace{55}} ,\\
\dim \left( V_{1}\right) &=&\underset{Q_{\alpha }}{\underbrace{32}}, \\
\dim \left( V_{2}\right) &=&\underset{P_{a}}{\underbrace{11}}+\underset{%
Z_{a_1 \cdots a_5}}{\underbrace{462}}=473.
\end{eqnarray}
The generators of the superalgebra underlying $D=11$ supergravity are given by the set
\begin{equation}
\lbrace \tilde{P}_a, \tilde{J}_{ab}, \tilde{Z}_{ab},\tilde{Z}_{a_1...a_5},\tilde{Q}_\alpha, \tilde{Q}'_\alpha \rbrace .
\end{equation}
These generators satisfy the following commutation relations:
\begin{eqnarray}
\lbrace \tilde{Q}, \tilde{\bar Q} \rbrace &=& -\left(\ii \Gamma^a \tilde{P}_a + \frac 12 \Gamma^{ab}\tilde{Z}_{ab}+ \frac {\ii}{5!} \Gamma^{a_1...a_5}\tilde{Z}_{a_1...a_5}\right)\,, \label{qq11}\\ \nonumber
\lbrace \tilde{Q}',\tilde{\bar Q}' \rbrace &=& 0\,,\\ \nonumber
\lbrace \tilde{Q},\tilde{\bar Q}' \rbrace &=& 0\,,\\ \nonumber
\left[\tilde{Q}, \tilde{P}_a\right] &=& -2 \ii \begin{pmatrix}
5 \\ 
0
\end{pmatrix}  \Gamma_a \tilde{Q}'\,,\\ \label{degpiqu} \nonumber
\left[\tilde{Q}, \tilde{Z}^{ab}\right] &=&-4  \Gamma^{ab}\tilde{Q}' \,, \\ \nonumber
\left[\tilde{Q}, \tilde{Z}^{a_1...a_5}\right] &=&- 2 \,(5!) \ii \begin{pmatrix}
\frac{1}{48} \\ 
\frac{1}{72}
\end{pmatrix}  \Gamma^{a_1...a_5}\tilde{Q}'\,, \\ \nonumber
\left[\tilde{J}_{ab}, \tilde{Z}^{cd}\right]&=&-8 \delta^{[c}_{[a}\tilde{Z}_{b]}^{\ d]}\,,\\ \nonumber
\left[\tilde{J}_{ab}, \tilde{Z}^{c_1\dots c_5}\right]&=&- 20 \delta^{[c_1}_{[a}\tilde{Z}^{c_2\dots c_5]}_{b]}\,,\\ \nonumber
\left[\tilde{J}_{ab}, \tilde{Q}\right]&=&- \Gamma_{ab} \tilde{Q}\,,\\ \nonumber
\left[\tilde{J}_{ab}, \tilde{Q}'\right]&=&- \Gamma_{ab} \tilde{Q}'\,, \\ \nonumber
\left[\tilde{P}_a , \tilde{Q}'\right] &=& \left[\tilde{Z}_{ab} , \tilde{Q}'\right]  = \left[\tilde{Z}_{a_1...a_5} , \tilde{Q}'\right] =
\left[\tilde{P}_a , \tilde{P}_b\right] =0, \\ \nonumber 
\left[ \tilde{J}^{ab},\tilde{P}_{c}\right] &=&\delta _{ec}^{ab}\tilde{P}^{e},  \\ \nonumber
\left[ \tilde{J}^{ab},\tilde{J}_{cd}\right] &=&\delta _{ecd}^{abf}\; \tilde{J}_{\; f}^{e} ,  \\ \nonumber
\left[ \tilde{Z}_{ab},\tilde{Z}_{bc}\right] &=& \left[ \tilde{Z}_{ab},\tilde{Z}_{a_1...a_5}\right] = \left[ \tilde{Z}_{ab},\tilde{P}_c \right] =\left[\tilde{P}_a, \tilde{Z}_{a_1...a_5}\right]=\left[\tilde{Z}_{a_1...a_5}, \tilde{Z}_{b_1...b_5}\right]=0,
\end{eqnarray}
where the free parameter $E_2$ appearing in Ref. \cite{Ravera} has been consistently fixed to the value $1$ (this is due to the possibility of fixing the normalization of the differential form associated with the extra fermionic generator $\tilde{Q}'$).

We observe that the above algebra actually describes two superalgebras, due to the degeneracy appearing in the commutation relation (\ref{degpiqu}), from which we clearly see that the generators $\tilde{Q}$ and $\tilde{P}_a$ can also commute. In the following, we will discuss the $S$-expansion, $0_S$-resonant-reduced procedure for both these superalgebras. 

Let us also observe that, in the description of the hidden superalgebra, the coefficients are written following the notation and conventions presented in Ref.s \cite{D'Auria, Ravera}, while, when considering the supersymmetric $osp(32/1)$ Lie algebra, we have adopted the notation presented in Ref. \cite{Iza1}. However, the coefficients appearing in the mentioned algebras are not relevant to our discussion, since we just need to know the \textit{structure} of the algebras for applying our analytic method. 

We can thus proceed, giving the internal decomposition of the target superalgebra (we first consider the case in which $\left[\tilde{Q}, \tilde{P}_a\right]\neq 0$). For the target superalgebra underlying $D=11$ supergravity, we can write 
\begin{eqnarray}
\dim \left( \tilde{V}_{0}\right) &=&\underset{\tilde{J}_{ab},\; \tilde{Z}_{ab}}{\underbrace{110}} ,\\
\dim \left( \tilde{V}_{1}\right) &=&\underset{\tilde{Q}_{\alpha },\; \tilde{Q}'_{\alpha }}{\underbrace{64}}, \\
\dim \left( \tilde{V}_{2}\right) &=&\underset{\tilde{P}_{a}}{\underbrace{11}}+\underset{%
\tilde{Z}_{a_1 \cdots a_5}}{\underbrace{462}}=473.
\end{eqnarray}
where we have clearly set $\tilde{V}_0=\lbrace 0\rbrace \cup \lbrace \tilde{J}_{ab},\tilde{Z}_{ab}\rbrace$, $\tilde{V}_1=\lbrace \tilde{Q},\tilde{Q}'\rbrace$, and $\tilde{V}_2=\lbrace \tilde{P}_a, \tilde{Z}_{a_1...a_5} \rbrace$. 
The subspaces partition for the target superalgebra satisfies the following relations
\begin{eqnarray}
\left[ \tilde{V}_{0},\tilde{V}_{0}\right] &\subset &\tilde{V}_{0} ,\\
\left[ \tilde{V}_{0},\tilde{V}_{1}\right] &\subset &\tilde{V}_0\oplus \tilde{V}_{1} ,\\
\left[ \tilde{V}_{0},\tilde{V}_{2}\right] &\subset &\tilde{V}_0\oplus \tilde{V}_{2} ,\\
\left[ \tilde{V}_{1},\tilde{V}_{1}\right] &\subset &\tilde{V}_{0}\oplus \tilde{V}_{2}, \\
\left[ \tilde{V}_{1},\tilde{V}_{2}\right] &\subset &\tilde{V}_{0}\oplus \tilde{V}_1 ,\\
\left[ \tilde{V}_{2},\tilde{V}_{2}\right] &\subset &\tilde{V}_{0}\oplus \tilde{V}_{2},
\end{eqnarray}
analogously to what we have done for $osp(32/1)$.
We can now move to the study of the usual system (\ref{system}), which, in this case, reads
\begin{equation}
\left\{
\begin{aligned}
& 110 =55(\tilde{P}-1-\Delta_1 -\Delta_2), \\
& 64 =32(\tilde{P}-1-\Delta_0 - \Delta_2), \\
& 473 =473(\tilde{P}-1-\Delta_0 - \Delta_1), \\
& \tilde{P} = \Delta_0 + \Delta_1 + \Delta_2 +1,
\end{aligned} \right.
\end{equation}
where $\Delta_0$, $\Delta_1$, $\Delta_2$ respectively denote the cardinality of the subsets related to the subspaces $V_0$, $V_1$, and $V_2$.
This system admits the unique solution
\begin{equation}
\tilde{P}=6, \;\;\; \Delta_0 = 2, \;\;\; \Delta_1 = 2, \;\;\; \Delta_2 = 1.
\end{equation}
Thus, we are now able to write the following subset decomposition of the set $\tilde{S}$ involved in the process:
\begin{eqnarray}
S_{2_0} &=& \lbrace \lambda_a, \lambda_b \rbrace , \\
S_{2_1} &=& \lbrace \lambda_c, \lambda_d \rbrace , \\
S_{1_2} &=& \lbrace \lambda_e  \rbrace . 
\end{eqnarray}
Now, we build up the adjoint representation with respect to the subspaces partition of the starting $osp(32/1)$ algebra, namely
\begin{equation}
\begin{aligned}
& \left( C\right) _{0B}^{C} =
\begin{pmatrix}
{(C)}_{00}^{0} & 0 & 0 \\ 
0 & {(C)}_{01}^{1} & 0 \\ 
0 & 0 & {(C)}_{02}^{2}
\end{pmatrix}, \;\;\; 
\left( C\right) _{1B}^{C} =
\begin{pmatrix}
0 & {(C)}_{10}^{1} & 0 \\ 
{(C)}_{11}^{0} & 0 & {(C)}_{11}^{2} \\ 
0 & {(C)}_{12}^{1} & 0
\end{pmatrix}, \nonumber \\
& \left( C\right) _{2B}^{C} =
\begin{pmatrix}
0 & 0 & {(C)}_{20}^{2} \\ 
0 & {(C)}_{21}^{1} & 0 \\ 
{(C)}_{22}^{0} & 0 & {(C)}_{21}^{2}
\end{pmatrix}.
\end{aligned}
\end{equation}
Thus, one can now write the usual relations (\ref{eqsemigrold}) (or their simpler form, given by (\ref{eqsemigr})) for the case under analysis, and find the following product structure for the subset decomposition of the set $\tilde{S}$
\begin{eqnarray}
S_{2_0} \cdot S_{2_0} &\subset& S_{2_0} \cup \lbrace \lambda_{0_S}\rbrace , \\
S_{2_0} \cdot S_{2_1}  &\subset& S_{2_1} \cup \lbrace \lambda_{0_S}\rbrace , \\
S_{2_0} \cdot S_{1_2}  &\subset& S_{1_2} \cup \lbrace \lambda_{0_S}\rbrace , \\
S_{2_1} \cdot S_{2_1} &\subset& \left( S_{2_0} \cap S_{1_2} \right) \cup \lbrace \lambda_{0_S}\rbrace , \\
S_{2_1} \cdot S_{1_2} &\subset& S_{1_2} \cup \lbrace \lambda_{0_S}\rbrace  , \\
S_{1_2} \cdot S_{1_2} &\subset& \left( S_{2_0} \cap S_{1_2} \right) \cup \lbrace \lambda_{0_S}\rbrace , 
\end{eqnarray}
where we have explicitly taken into account the presence of the zero element $\lambda_{0_S}$. This allows to reach the multiplication rules
\begin{eqnarray}
\lambda_{a,b} \lambda_{a,b} &=& \lambda_{a,b,0_S}, \\
\lambda_{a,b} \lambda_{c,d} &=& \lambda_{c,d,0_S}, \\
\lambda_{a,b} \lambda_{e} &=& \lambda_{e,0_S}, \\
\lambda_{c,d} \lambda_{c,d} &=& \lambda_{a,b,e,0_S}, \\
\lambda_{c,d} \lambda_{e} &=& \lambda_{e,0_S}, \\
\lambda_{e} \lambda_{e} &=& \lambda_{a,b,e,0_S},
\end{eqnarray}
where we have already taken into account the triviality of the multiplications rules
\begin{align}
\lambda_{0_S}\lambda_{0_S}= &\lambda_{0_S}, \\
\lambda_{0_S}\lambda_{a,b,c,d,e}= &\lambda_{0_S}. 
\end{align}

We can now fix the degeneracy appearing in the above multiplication rules, by analyzing the information coming from the target superalgebra. 
According to the usual $S$-expansion procedure (see Ref. \cite{Iza1}), we have to write the commutation relations between the generators of the target superalgebra in terms of the commutation relations between the generators of the $S$-expanded $osp(32/1)$. 
After having performed the identification (according to the identification criterion presented in Subsection \ref{IdCr})
\begin{equation}
\lambda_a J_{ab} = \tilde{J}_{ab}, \;\;\; \lambda_b J_{ab} = \tilde{Z}_{ab}, \;\;\; \lambda_c Q = \tilde{Q}, \;\;\; \lambda_d Q = \tilde{Q}', \;\;\; \lambda_e P_a = \tilde{P}_a, \;\;\; \lambda_e Z_{a_1...a_5} = \tilde{Z}_{a_1...a_5},
\end{equation}
we are able to write the commutation relations of the superalgebra underlying $D=11$ supergravity in terms of the commutation relations of the $S$-expanded generators of $osp(32/1)$. 
In the following, we will just consider the structure of the commutation relations, since the explicit values of the coefficients are not relevant to our analysis. 
For performing this calculation, we consider the case in which $\left[\tilde{Q}, \tilde{P}_a\right]\neq 0$. 

Thus, taking into account the commutation relations for the initial algebra $osp(32/1)$, we have:
\begin{eqnarray}
& \lbrace \tilde{Q}', \tilde{Q}' \rbrace = \lbrace\lambda_d Q, \lambda_d Q \rbrace=\lambda_d \lambda_d \lbrace Q,Q \rbrace = 0 \;\;\; \rightarrow \;\;\; \lambda_d \lambda_d = \lambda_{0_S} , \\
& \lbrace \tilde{Q}, \tilde{Q}' \rbrace = \lbrace\lambda_c Q, \lambda_d Q \rbrace=\lambda_c \lambda_d \lbrace Q,Q \rbrace = 0 \;\;\; \rightarrow \;\;\; \lambda_c \lambda_d = \lambda_{0_S} , \\
& \left[\tilde{P}_a, \tilde{P}_b \right] = \left[\lambda_e P_a, \lambda_e P_b \right]=\lambda_e \lambda_e \left[P_a,P_b \right]= 0 \;\;\; \rightarrow \;\;\; \lambda_e \lambda_e = \lambda_{0_S} , \\
& \left[\tilde{P}_a, \tilde{J}_{bc} \right] = \left[\lambda_e P_a, \lambda_a J_{bc} \right]= \lambda_e \lambda_a \left[P_a,J_{bc} \right]\propto \lambda_e \delta _{ec}^{ab}P^{e} \;\;\; \rightarrow \;\;\; \lambda_e \lambda_a = \lambda_e , \\
& \left[\tilde{J}_{ab}, \tilde{J}_{cd} \right] = \left[\lambda_a J_{ab}, \lambda_a J_{cd} \right]= \lambda_a \lambda_a \left[J_{ab},J_{cd} \right]\propto \lambda_a \delta _{ecd}^{abf}\; J_{\; f}^{e} \;\;\; \rightarrow \;\;\; \lambda_a  \lambda_a = \lambda_a , \\
& \left[\tilde{P}_a, \tilde{Q}' \right] = \left[\lambda_e P_a, \lambda_d Q \right]=\lambda_e \lambda_d \left[P_a,Q \right]= 0 \;\;\; \rightarrow \;\;\; \lambda_e \lambda_d = \lambda_{0_S} , \\
& \left[\tilde{Z}_{ab}, \tilde{Q}' \right] = \left[\lambda_b J_{ab}, \lambda_d Q \right]=\lambda_b \lambda_d \left[J_{ab},Q \right]= 0 \;\;\; \rightarrow \;\;\; \lambda_b \lambda_d = \lambda_{0_S} , \\
& \left[\tilde{P}_a, \tilde{Q} \right] = \left[\lambda_e P_a, \lambda_c Q \right]=\lambda_e \lambda_c \left[P_a,Q \right]\propto \lambda_d Q \;\;\; \rightarrow \;\;\; \lambda_e \lambda_c = \lambda_d , \\
& \left[\tilde{Z}_{ab}, \tilde{Q} \right] = \left[\lambda_b J_{ab}, \lambda_c Q \right]=\lambda_b \lambda_c \left[J_{ab},Q \right]\propto \lambda_d Q \;\;\; \rightarrow \;\;\; \lambda_b \lambda_c = \lambda_d , \\
& \left[\tilde{J}_{ab}, \tilde{Z}_{cd} \right] = \left[\lambda_a J_{ab}, \lambda_b J_{cd} \right]= \lambda_a \lambda_b \left[J_{ab},Z_{cd} \right]\propto \lambda_b \delta _{ecd}^{abf}\; Z_{\; f}^{e} \;\;\; \rightarrow \;\;\; \lambda_a \lambda_b = \lambda_b , \\
& \left[\tilde{J}_{ab}, \tilde{Q} \right] = \left[\lambda_a J_{ab}, \lambda_c Q \right]=\lambda_a \lambda_c \left[J_{ab},Q \right]\propto \lambda_c Q \;\;\; \rightarrow \;\;\; \lambda_a \lambda_c = \lambda_c , \\
& \left[\tilde{J}_{ab}, \tilde{Q}' \right] = \left[\lambda_a J_{ab}, \lambda_d Q \right]=\lambda_a \lambda_d \left[J_{ab},Q \right]\propto \lambda_d Q \;\;\; \rightarrow \;\;\; \lambda_a \lambda_d = \lambda_d , \\
& \left[\tilde{Z}_{ab}, \tilde{Z}_{cd} \right] = \left[\lambda_b J_{ab}, \lambda_b J_{cd} \right]= \lambda_b \lambda_b \left[J_{ab},Z_{cd} \right]=0 \;\;\; \rightarrow \;\;\; \lambda_b \lambda_b = \lambda_{0_S} , \\
& \left[\tilde{Z}_{ab}, \tilde{P}_{c} \right] = \left[\lambda_b J_{ab}, \lambda_e P_{c} \right]= \lambda_b \lambda_e \left[J_{ab},P_{c} \right]=0 \;\;\; \rightarrow \;\;\; \lambda_b \lambda_e = \lambda_{0_S} , \\
& \lbrace \tilde{Q}, \tilde{Q} \rbrace = \lbrace\lambda_c Q, \lambda_c Q \rbrace=\lambda_c \lambda_c \lbrace Q,Q \rbrace \propto \lambda_e P_a + \lambda_b J_{ab}+\lambda_e Z_{abcde} \; \rightarrow  \label{lastcomm} \\
& \rightarrow \; \lambda_c \lambda_c = \lambda_b , \; \; \text{where we have set} \; \lambda_b=\lambda_e,  \nonumber
\end{eqnarray}
and the other commutation relations give us results that agree with the above ones. 

We observe that, in equation (\ref{lastcomm}), we must set 
\begin{equation}\label{equal}
\lambda_b=\lambda_e ,
\end{equation}
in order to get consistent relations. 
For performing this identification with consistency, we have also exploited the statement which follows from Theorem \ref{T}.

This procedure fixes the degeneracy of the multiplication rules between the elements of the subsets of $\tilde{S}$, and we are now able to write the following multiplication table
\begin{equation}
\begin{array}{c|ccccc}
& \lambda _{a} & \lambda _{b} & \lambda _{c} & \lambda _{d} & \lambda_{0_S} \\ 
\hline
\lambda _{a} & \lambda _{a} & \lambda _{b} & \lambda _{c} & \lambda _{d} & \lambda_{0_S}
\\ 
\lambda _{b} & \lambda _{b} & \lambda _{0_S} & \lambda _{d} & \lambda _{0_S} & \lambda_{0_S}
\\ 
\lambda _{c} & \lambda _{c} & \lambda _{d} & \lambda _{b} & \lambda _{0_S} & \lambda_{0_S}
\\ 
\lambda _{d} & \lambda _{d} & \lambda _{0_S} & \lambda _{0_S} & \lambda _{0_S} & \lambda_{0_S}
\\
\lambda _{0_S} & \lambda _{0_S} & \lambda _{0_S} & \lambda _{0_S} & \lambda _{0_S} & \lambda_{0_S}
\end{array}
\end{equation}
Then, after having performed the identification
\begin{equation}
a \leftrightarrow  0, \;\;\; b \leftrightarrow  2, \;\;\; c \leftrightarrow  1, \;\;\;
d \leftrightarrow  3, \;\;\; 0_S \leftrightarrow  4,
\end{equation} 
we can finally rewrite the table above as follows (in the usual order):
\begin{equation}
\begin{array}{c|ccccc}
& \lambda _{0} & \lambda _{1} & \lambda _{2} & \lambda _{3} & \lambda_{4} \\ 
\hline
\lambda _{0} & \lambda _{0} & \lambda _{1} & \lambda _{2} & \lambda _{3} & \lambda_{4}
\\ 
\lambda _{1} & \lambda _{1} & \lambda _{2} & \lambda _{3} & \lambda _{4} & \lambda_{4}
\\ 
\lambda _{2} & \lambda _{2} & \lambda _{3} & \lambda _{4} & \lambda _{4} & \lambda_{4}
\\ 
\lambda _{3} & \lambda _{3} & \lambda _{4} & \lambda _{4} & \lambda _{4} & \lambda_{4}
\\
\lambda _{4} & \lambda _{4} & \lambda _{4} & \lambda _{4} & \lambda _{4} & \lambda_{4}
\end{array}
\end{equation}
This is exactly the multiplication table of the semigroup $S^{(3)}_E$, which leads, as it was shown in Ref. \cite{Iza1}) through a $S$-expansion procedure ($0_S$-resonant-reduction), from the $osp(32/1)$ algebra to the hidden superalgebra underlying the eleven-dimensional supergravity theory, described in Ref.s \cite{D'Auria, Ravera}.


\begin{thebibliography}{99}

\bibitem{Hatsuda} M. Hatsuda, M. Sakaguchi, \textit{Wess-Zumino term for the AdS superstring and generalized Inonu-Wigner contraction}, Prog. Theor. Phys. \textbf{109} (2003) 853. arXiv:0106114 [hep-th]. 

\bibitem{Azca1} J.A. de Azc\'arraga, J.M. Izquierdo, M. Pic\'{o}n and O. Varela, \textit{Generating Lie and gauge free differential (super)algebras by expanding Maurer-Cartan forms and Chern-Simons supergravity}, Nucl. Phys. B \textbf{662}, 185-219 (2003). arXiv:0212347 [hep-th]. 

\bibitem{Azca2} J.A. de Azc\'arraga, J.M. Izquierdo, M. Pic\'{o}n, O. Varela, \textit{Extensions, expansions, Lie algebra cohomology and enlarged superspaces}, Class. Quantum Gravity \textbf{21} (2004) S1375-1384. arXiv:0401033 [hep-th]. 

\bibitem{Azca3} J.A. de Azc\'arraga, J.M. Izquierdo, M. Pic\'{o}n, O. Varela, \textit{Expansions of algebras and superalgebras and some applications}, Int. J. Theor. Phys. \textbf{46} (2007) 2738. arXiv:0703017 [hep-th]. 

\bibitem{Iza1} F. Izaurieta, E. Rodr\'iguez, P. Salgado, \textit{Expanding Lie (super)algebras through Abelian semigroups}, J. Math. Phys. 
\textbf{47} (2006) 123512. arXiv:0606215 [hep-th]. 

\bibitem{Iza2} F. Izaurieta, E. Rodr\'iguez, P. Salgado, \textit{Construction of Lie algebras and invariant tensors through abelian semigroups}, J. Phys. Conf. Ser. \textbf{134} (2008) 012005. 

\bibitem{Iza3} Fernando Izaurieta, Eduardo Rodr\'iguez, \textit{Dual Formulation of the Lie Algebra S-expansion Procedure}, J. Math. Phys. \textbf{50} (2009) 073511. arXiv:0903.4712 [hep-th]. 

\bibitem{Iza4} Fernando Izaurieta, Eduardo Rodr\'iguez, Patricio Salgado, \textit{Eleven-Dimensional Gauge Theory for the M Algebra as an Abelian Semigroup Expansion of $osp(32/1)$}, Eur. Phys. J. C \textbf{54} (2008) 675-684. arXiv:0606225 [hep-th]. 

\bibitem {GRCS} F. Izaurieta, P. Minning, A. Perez, E. Rodr\'{\i}guez, P.
Salgado, \textit{Standard General Relativity from Chern-Simons Gravity}, Phys. Lett. B \textbf{678} 213 (2009). arXiv:0905.2187 [hep-th]. 

\bibitem {CPRS1} P.K. Concha, D.M. Pe\~{n}afiel, E.K. Rodr\'{\i}guez, P.
Salgado, \textit{Even-dimensional General Relativity from Born-Infeld
gravity}, Phys. Lett. B \textbf{725}, 419 (2013). arXiv:1309.0062 [hep-th]. 

\bibitem {Topgrav} P. Salgado, R. J. Szabo, O. Valdivia, \textit{Topological
gravity and transgression holography}, Phys. Rev. D \textbf{89} (2014) 084077.
arXiv:1401.3653 [hep-th]. 

\bibitem {BDgrav} C. Inostroza, A. Salazar, P. Salgado, \textit{Brans-Dicke
gravity theory from topological gravity}, Phys. Lett. B \textbf{734} (2014) 377. 

\bibitem {CR2} P.K. Concha, E.K. Rodr\'{\i}guez, \textit{N=1 supergravity and
Maxwell superalgebras}, JHEP \textbf{1409} (2014) 090. arXiv:1407.4635 [hep-th]. 

\bibitem {CRSnew} P.K. Concha, E.K. Rodr\'{\i}guez, P. Salgado, \textit{Generalized supersymmetric cosmological term in $N=1$ supergravity}, JHEP \textbf{08} (2015) 009. arXiv:1504.01898 [hep-th]. 

\bibitem{Static} Juan Cris\'ostomo, Fernando Gomez, Patricio Mella, Cristian Quinzacara, Patricio Salgado, \textit{Static solutions in Einstein-Chern-Simons gravity}, Journal of Cosmology and Astroparticle Physics, Issue 06, article id. 049 (2016). 

\bibitem{Gen} P.K. Concha, D.M. Pe\~nafiel, E.K. Rodr\'iguez, P. Salgado, \textit{Generalized Poincar\'e algebras and Lovelock-Cartan gravity theory}, Phys. Lett. B \textbf{742} (2015) 310-316. arXiv:1405.7078 [hep-th]. 

\bibitem{Ein} N.L. Gonz\'{a}lez Albornoz, P. Salgado, G. Rubio, S. Salgado, \textit{Einstein-Hilbert action with cosmological term from Chern-Simons gravity}, J. Geom. Phys. \textbf{86} (2014) 339. arXiv:1605.00325 [math-ph]. 

\bibitem{Fierro1} O. Fierro, F. Izaurieta, P. Salgado, O. Valdivia,
\textit{(2+1)-dimensional supergravity invariant under the AdS-Lorentz
superalgebra}, arXiv:1401.3697 [hep-th]. 

\bibitem{Fierro2} Jos\'e D\'iaz, Octavio Fierro, Fernando Izaurieta, Nelson Merino, Eduardo Rodr\'iguez, Patricio Salgado, Omar Valdivia, \textit{A generalized action for (2+1)-dimensional Chern-Simons gravity}, J.Phys. A \textbf{45} (2012) 255207. 




\bibitem{Knodrashuk} R. Caroca, I. Knodrashuk, N. Merino, F. Nadal, \textit{Bianchi spaces and their 3-dimensional isometries as S-expansions of 2-dimensional isometries}, J. Phys. A: Math. Theor. \textbf{46} (2013) 225201. 

\bibitem{Artebani} M. Artebani, R. Caroca, M. C. Ipinza, D. M. Pe\~{n}afiel,
P. Salgado, \textit{Geometrical aspects of the Lie Algebra S-expansion
Procedure}, J. Math. Phys. Vol. \textbf{57} Issue 2 (2016). arXiv:1602.0452 [hep-th]. 

\bibitem{Concha1} P.K. Concha, D.M. Pe\~{n}afiel, E.K. Rodr\'{\i}guez, P.
Salgado, \textit{Chern-Simons and Born-Infeld gravity theories and Maxwell
algebras type}, Eur. Phys. J. C \textbf{74} (2014) 2741. arXiv:1402.0023 [hep-th]. 

\bibitem{Salgado} P. Salgado, S. Salgado, \textit{$\mathfrak{so} (D - 1, 1)
\otimes \mathfrak{so} (D - 1, 2)$ algebras and gravity}, Phys. Lett. B 
\textbf{728} 5 (2014). 

\bibitem{Concha2} P.K. Concha, E.K. Rodr\'{\i}guez, \textit{Maxwell
Superalgebras and Abelian Semigroup Expansion}, Nucl. Phys. B \textbf{886},
1128-1152 (2014). arXiv:1405.1334 [hep-th]. 


\bibitem{Caroca:2010ax} R.~Caroca, N.~Merino and P.~Salgado, \textit{S-Expansion of Higher-Order Lie Algebras}, J.\ Math.\ Phys.\ {\bf 50}, 013503 (2009). arXiv:1004.5213 [math-ph]. 

\bibitem{Caroca:2010kr} R.~Caroca, N.~Merino, A.~Perez and P.~Salgado, \textit{Generating Higher-Order Lie Algebras by Expanding Maurer Cartan Forms}, J.\ Math.\ Phys.\ {\bf 50}, 123527 (2009). arXiv:1004.5503 [hep-th]. 

\bibitem{Caroca:2011zz} R.~Caroca, N.~Merino, P.~Salgado and O.~Valdivia, \textit{Generating infinite-dimensional algebras from loop algebras by expanding Maurer-Cartan forms}, J.\ Math.\ Phys.\ {\bf 52}, 043519 (2011). 

\bibitem{Andrianopoli:2013ooa} L.~Andrianopoli, N.~Merino, F.~Nadal and M.~Trigiante, \textit{General properties of the expansion methods of Lie algebras}, J.\ Phys.\ A {\bf 46}, 365204 (2013). arXiv:1308.4832 [gr-qc]. 

\bibitem{Concha:2016hbt} P.~K.~Concha, R.~Durka, N.~Merino and E.~K.~Rodríguez, \textit{New family of Maxwell like algebras}, Phys.\ Lett.\ B {\bf 759}, 507 (2016). arXiv:1601.06443 [hep-th]. 

\bibitem{Concha:2016kdz} P.~K.~Concha, R.~Durka, C.~Inostroza, N.~Merino and E.~K.~Rodríguez, \textit{Pure Lovelock gravity and Chern-Simons theory}, Phys.\ Rev.\ D {\bf 94}, no. 2, 024055 (2016). arXiv:1603.09424 [hep-th]. 

\bibitem{Concha:2016tms} P.~K.~Concha, N.~Merino and E.~K.~Rodríguez, \textit{Lovelock gravities from Born-Infeld gravity theory}. arXiv:1606.07083 [hep-th]. 

\bibitem{Durka:2016eun} R.~Durka, \textit{Resonant algebras and gravity}. arXiv:1605.00059 [hep-th]. 

\bibitem{Inonu1} E. In\"on\"u and E.P. Wigner, \textit{On the contraction of groups and their representations},
Proc. Nat. Acad. Sci. U.S.A. \textbf{39}, 510-524 (1953). 

\bibitem{Inonu2} E. In\"on\"u, \textit{Contractions of Lie groups
and their representations, in Group theoretical concepts in elementary particle physics}, F.G\"ursey ed., Gordon and Breach, 391-402 (1964). 

\bibitem{gaussbonnet} P.K. Concha, M. C. Ipinza, L. Ravera, E. K. Rodr\'{i}guez, \textit{On the Supersymmetric Extension of Gauss Bonnet like Gravity}, JHEP \textbf{1609} (2016) 007. arXiv:1607.00373 [hep-th]. 

\bibitem{Cremmer} E. Cremmer, B. Julia and J. Scherk, \textit{Supergravity
Theory in Eleven-Dimensions}, Phys. Lett. B \textbf{76} (1978) 409.
doi:10.1016/0370-2693(78)90894-8 

\bibitem{Cremmer2} E. Cremmer and B. Julia, Phys. Lett. B \textbf{80} (1978) 48. 

\bibitem{D'Auria} R. D'Auria and P. Fr\'e, \textit{Geometric Supergravity in
d = 11 and Its Hidden Supergroup}, Nucl. Phys. B \textbf{201} (1982) 101
Erratum: [Nucl. Phys. B \textbf{206} (1982) 496].
doi:10.1016/0550-3213(82)90376-5 

\bibitem{Ravera} L. Andrianopoli, R. D'Auria, L. Ravera, \textit{Hidden gauge structure of supersymmetric free differential algebras}, JHEP \textbf{1608} (2016) 095. arXiv:1606.07328 [hep-th]. 

\bibitem{Hull} C.M. Hull and P.K. Townsend, \textit{Unity of superstring dualities}, Nucl. Phys. B \textbf{438} (1995) 109. arXiv:9410167 [hep-th]. 

\bibitem{Townsend} P.K. Townsend, \textit{P-brane democracy}, In *Duff, M.J. (ed.): The world in eleven dimensions* 375-389. arXiv:9507048 [hep-th]. 

\bibitem{vanHolten} J.W. van Holten and A. Van Proeyen, \textit{$\mathcal{N}=1$ Supersymmetry Algebras in $D=2$, $D=3$, $D=4$ MOD-$8$}, J. Phys. A \textbf{15} (1982) 3763. doi:10.1088/0305-4470/15/12/028 
  
\bibitem{Bandos} I.A. Bandos, J. A. de Azc\'{a}rraga, J.M. Izquierdo, M. Pic\'{o}n and O. Varela, \textit{On the underlying gauge group structure of D=11 supergravity}, Phys. Lett. B \textbf{596} (2004) 145. arXiv:0406020 [hep-th]. 
\\
I.A. Bandos, J.A. de Azc\'{a}rraga, M. Pic\'{o}n and O. Varela, \textit{On the formulation of D = 11 supergravity and the composite nature of its three-form gauge field}, Annals Phys. \textbf{317} (2005) 238. arXiv:0409100 [hep-th]. 

\bibitem{Witten} E. Witten, \textit{2 + 1 Dimensional Gravity as an Exactly Soluble System}, Nucl. Phys. B \textbf{311} (1988/89) 46-78. 

\end{thebibliography}
\end{document}